\documentclass[11pt]{article}
\usepackage{amsmath,amsthm,amsfonts,authblk}
\usepackage{amssymb,graphicx,verbatim,multirow}
\usepackage[round]{natbib}
\usepackage{color}

\definecolor{darkgreen}{rgb}{0,0.4,0}

\def\d{\delta}

\def\d{\delta}

\newtheorem{theorem}{Theorem}
\newtheorem{proposition}[theorem]{Proposition}
\newtheorem{corollary}[theorem]{Corollary}
\newtheorem{lemma}[theorem]{Lemma}

\title{Piecewise linear and Boolean models of chemical reaction networks}
\author[a,b]{Alan Veliz-Cuba\footnote{Correspondence Author\\
Alan Veliz-Cuba (\texttt{alanavc@math.uh.edu}),
 Ajit Kumar (\texttt{ajit.kumar@snu.edu.in}), Kre\v{s}imir Josi\'{c} (\texttt{josic@math.uh.edu}) 
 }
 }

\author[a,c]{Ajit Kumar}

\author[a,d]{Kre\v{s}imir Josi\'{c}}

\affil[a]{Department of Mathematics, University of Houston, Houston, Texas 77204-3008, USA}
\affil[b]{Department of Biochemistry and Cell Biology, Rice University, Houston, Texas 77251-1892, USA}
\affil[c]{Department of Mathematics, Shiv Nadar Univerisity, Greater Noida, Uttar Pradesh 203207, India}
\affil[d]{Department of Biology and Biochemistry, University of Houston, Houston, Texas 77204-5001, USA}

\begin{document}
\maketitle
\abstract{Models of biochemical networks are frequently high-dimensional and complex. Reduction 
methods that preserve important dynamical properties are therefore essential in their study. 
Interactions between the nodes in such networks are frequently modeled using a Hill function, $x^n/(J^n+x^n)$.  Reduced  ODEs and Boolean networks have been studied extensively when  
the exponent $n$ is large. 
However, the case of small constant $J$ appears in practice, but is not well understood. In this paper we provide a mathematical analysis 
of this limit, and show that a reduction to a set of piecewise linear ODEs and Boolean networks can be mathematically justified. The piecewise linear systems have closed form solutions that closely track those of the fully nonlinear model.  On the other hand, the simpler, Boolean network can be used to study the qualitative behavior of the original system. We justify the reduction using geometric singular perturbation theory and compact convergence, and illustrate the results in networks modeling a genetic switch and a genetic oscillator.
}

\section{Introduction}

Accurately describing the behavior of interacting enzymes, proteins, and genes
requires spatially extended stochastic models.  However, such models are
difficult to implement and fit to data. Hence tractable reduced models are frequently used instead. 
In many popular models of biological networks, a single ODE is used 
to describe each node, and sigmoidal functions to describe interactions between them.  Even such simplified ODEs are typically intractable, as the number of parameters and
 the potential dynamical complexity make it difficult to analyze the behavior of the system using purely  numerical methods.
   Reduced models that capture the overall dynamics, or allow approximate solutions
can be of great help in this situation [\cite{Verhulst_GSPT,Hek2010}].

Analytical treatments are possible in certain limits.  The
approaches that have been developed  to analyze models of gene interaction
networks can be broadly classified into three
categories~[\cite{PolynikisHoganBernardo2009}]:  \emph{Quasi Steady State Approximations} (QSSA), \emph{Piecewise Linear Approximations} (PLA), and
 \emph{discretization of continuous time
ODEs}. In particular, in certain limits interactions between network elements become
switch--like~[\cite{kauffman69,Snoussi1989,Mochizuki2005, Alon2006,mendoza2006, DavitichBornholdt2008, Wittmann2009, Franke2010, Veliz-CubaArthurHochstetlerKlompsKorpi2012}].   For instance,
the Hill function,   $f(x) = x^n / (x^n + J^n )$,  approaches the Heaviside function, $H(x - J)$,
in the limit of large $n$.
In this limit the domain on which the network is modeled
is naturally split into subdomains: The threshold,
corresponding to the parameter $J$ in the Hill function, divides the domain into two subdomains within which
the Heaviside function is constant.  Within each subdomain a node is either
fully expressed, or
not expressed at all.  When $n$ is large, the Hill function, $f(x)$, is approximately constant in each
of the subdomains, and boundary layers occur when $x$ is
close to the threshold, $x \approx J$~[\cite{IroniPanzeriPlahteSimoncini2011}]. To simplify the system further, we can map values of $x$ below the threshold to 0, and the values above the threshold to 1 to obtain a Boolean network (BN); that is, a map
\[h=(h_1,\ldots,h_N):\{0,1\}^N\rightarrow \{0,1\}^N,\]
where each function $h_i$ describes how variable $i$ qualitatively depends on the other variables
 [\cite{GlassKauffman1973,Snoussi1989,thomasbook,edwards2000,edwards2001}].  Such reduced systems are simpler to analyze, and share the  dynamical properties of the original system, if the reduction is done properly.

The reduced models obtained in the limit of a large Hill coefficient, $n,$ have a long and rich history. Piecewise linear
functions of the form proposed in~[\cite{GlassKauffman1973}] have been shown to
be well suited for the modeling of genetic regulatory networks, and can sometimes be justified rigorously [\cite{de04}].  In particular,  singular perturbation
theory  can
be used to obtain reduced equations within each subdomain and the boundary
layers,
and global approximations   within the entire
domain~[\cite{IroniPanzeriPlahteSimoncini2011}]. On the other hand, although BNs have been used to model the dynamics of different biological systems, their relation to more complete models
was mostly demonstrated with case studies, heuristically or only for steady states [\cite{GlassKauffman1973,Snoussi1989, thomasbook, albert2003, mendoza2006, DavitichBornholdt2008, p53, p53ode, Wittmann2009, Franke2010, Veliz-CubaArthurHochstetlerKlompsKorpi2012}].

Here we again start
with the Hill function, $ x^n / (x^n + J^n )$, but instead of assuming that $n$ is large, we assume that $J$ is small. This case
 has a simple physical
interpretation: Consider the Hill function that occurs in the Michaelis-Menten
scheme, which models the catalysis of the inactive form of some
protein to its active form in the presence of an enzyme. When $J$ is small
the total enzyme concentration is much smaller than the total protein
concentration.
Although the subsequent results hold for any fixed $n$, for simplicity we
assume $n = 1$.

More precisely, we consider a model biological 
network where the activity at each of $N$ nodes is described by $u_i \in [0,1]$, and evolves according to
\begin{align}\label{eqn:ProblemEquation}
	\frac{du_i}{dt} = A_i\frac{1-u_i}{J_{i}^A+1-u_i}-I_i\frac{u_i}{J_{i}^I+u_i},
\end{align}
where $J_i^A, J_i^I>0$, and the functions $A_i=A_i(u)$, $I_i=I_i(u)$ are affine functions.

This type of equations have been used successfully in many models [\cite{GoldbeterKoshland1981, Goldbeter1991, NovakPatakiCilibertoTyson2001, de02, NovakPatakiCilibertoTyson2003, IshiiSugaHagiyaWatanabeMoriYoshinoTomita2007, CilibertoFabrizioTyson2007, DavitichBornholdt2008, vanZwietenRoodaArmbrusterNagy2011}]. Here  $A_i$ and $I_i$  describe how the other variables affect $u_i$  and can represent activation/phosphorylation/ production and inhibition/dephosphorylation/decay, respectively. The variables $u_i$ can represent species such as protein concentrations, the active form of enzymes, or activation level of genes.  A simple example is provided by a  protein that can exist in an unmodified form, $W,$ and a modified form, $W^*,$ (e.g. proteases, and Cdc2, Cdc25, Wee1, and Mik1 kinases [\cite{Goldbeter1991, Novak1998, NovakPatakiCilibertoTyson2001}]) where the conversion between the two forms is catalyzed by two enzymes, $E_1$ and $E_2$ [\cite{GoldbeterKoshland1981, Goldbeter1991,Novak1998, NovakPatakiCilibertoTyson2001}] (See Appendix for details). 
However, note that the models of chemical reactions we consider can be rigorously derived from the Chemical Master Equation only in the case of a single reaction~[\cite{KumarJosic2011}]. The models of networks of chemical reactions that we take as the starting point of 
our reduction should therefore be regarded as phenomenological. 

It is easy to show that the region $0\leq u_i\leq 1, 1 \leq i \leq N$ is invariant so that Eq.~\eqref{eqn:ProblemEquation} is a system of equations on $[0,1]^N$. Equations involving this special class of Hill functions are generally referred to as
Michaelis-Menten type equations, and $J$  the Michaelis-Menten 
constant~[\cite{MichaelisMenten1913, GoldbeterKoshland1981, Goldbeter1991, NovakTyson1993, NovakPatakiCilibertoTyson2001, NovakPatakiCilibertoTyson2003, CilibertoFabrizioTyson2007, DavitichBornholdt2008, ChaoTang2009}]. 

The  constants $J$ are frequently very small in practice [\cite{NovakPatakiCilibertoTyson2001,DavitichBornholdt2008}], which motivates examining Eq.~\eqref{eqn:ProblemEquation} 
when  $0 < J \ll 0$. In this case, we discuss a two step reduction of the model 
$$
\text{full, nonlinear model $\longrightarrow$ piecewise linear model  (PL) $\longrightarrow$ Boolean Network (BN)}.
$$ 
We first illustrate this reduction using two standard examples, and then provide a general mathematical justification. We note  that the reduction obtained in the first step (see Eq.~\eqref{eqn:mainReduced_RSTa}) is actually (algebraic) piecewise affine.
However, it is customary  to refer to the equation and the associated model as \emph{piecewise linear} [\cite{GlassKauffman1973,Snoussi1989, thomasbook,edwards2000,de02}], and we follow this convention.

The main idea behind the piecewise linear (PL) reduction is simple:
If $J \ll x$ then the Hill functions, $f(x) = x / (x + J ) \approx 1$. However, when $x$ and $J$ are comparable, $x \sim J,$ this is no longer true.  In this boundary layer, we  rescale variables by introducing $\tilde{x} := x/J$. A similar argument works for the function $(1-x) / (J+1-x )$ (see Appendix). We show that using this observation, the domain $[0,1]^N$ naturally decomposes  into a nested sequence of hypercubes.  The dynamics on each hypercube in the sequence is described by a solvable differential-algebraic system of equations.  The PL reduction therefore gives an \emph{analytically tractable} approximate
solution to the original system.

In the next step of the reduction we obtain a Boolean Network (BN): The PL  approximation is used to divide $[0,1]^N$ into chambers.   Within nearly all of a chamber the rate of change of each element of the network is constant when $J \ll 1$.   We use these chambers to define a BN. A similar approach was recently used to motivate a Boolean reduction of a model protein interaction
network~[\cite{DavitichBornholdt2008}].

The mathematical justification also follows two steps.  We use Geometric Singular Perturbation Theory (GSPT) in Section \ref{sec:math_PL} to justify the PL approximation. 
The justification of the BN reduction is given in Section \ref{sec:math_BN}.
We show that  there is a one-to-one correspondence between steady states (equilibrium solutions) of the BN and the full and PL system near the vertices of $[0,1]^N$. Futhermore, we show that this  one-to-one correspondence between steady states is actually global (up to a set of small measure in $[0,1]^N$). BNs have been used to study oscillatory behavior [\cite{Li_cc_2004,p53}], and we prove in Section \ref{sec:math_BN_trajectories} that under some conditions oscillations in a BN correspond to oscillations in the full system.

\section{Example problems}\label{ExampleProblems}

We start by demonstrating the main idea of our approach using networks of
two and three mutually repressing nodes. These nodes
can represent genes that mutually inhibit each other's
production~[\cite{GardnerCantorCollins2000,ElowitzLeibler2000}].  
However, 
the theory we develop applies whenever the heuristic model given in Eq.~\eqref{eqn:ProblemEquation} is applicable.
We accompany these examples with a 
heuristic explanation of the different steps in the
reduction. 

\begin{figure}[t]
\begin{center}
\includegraphics[scale = .35]{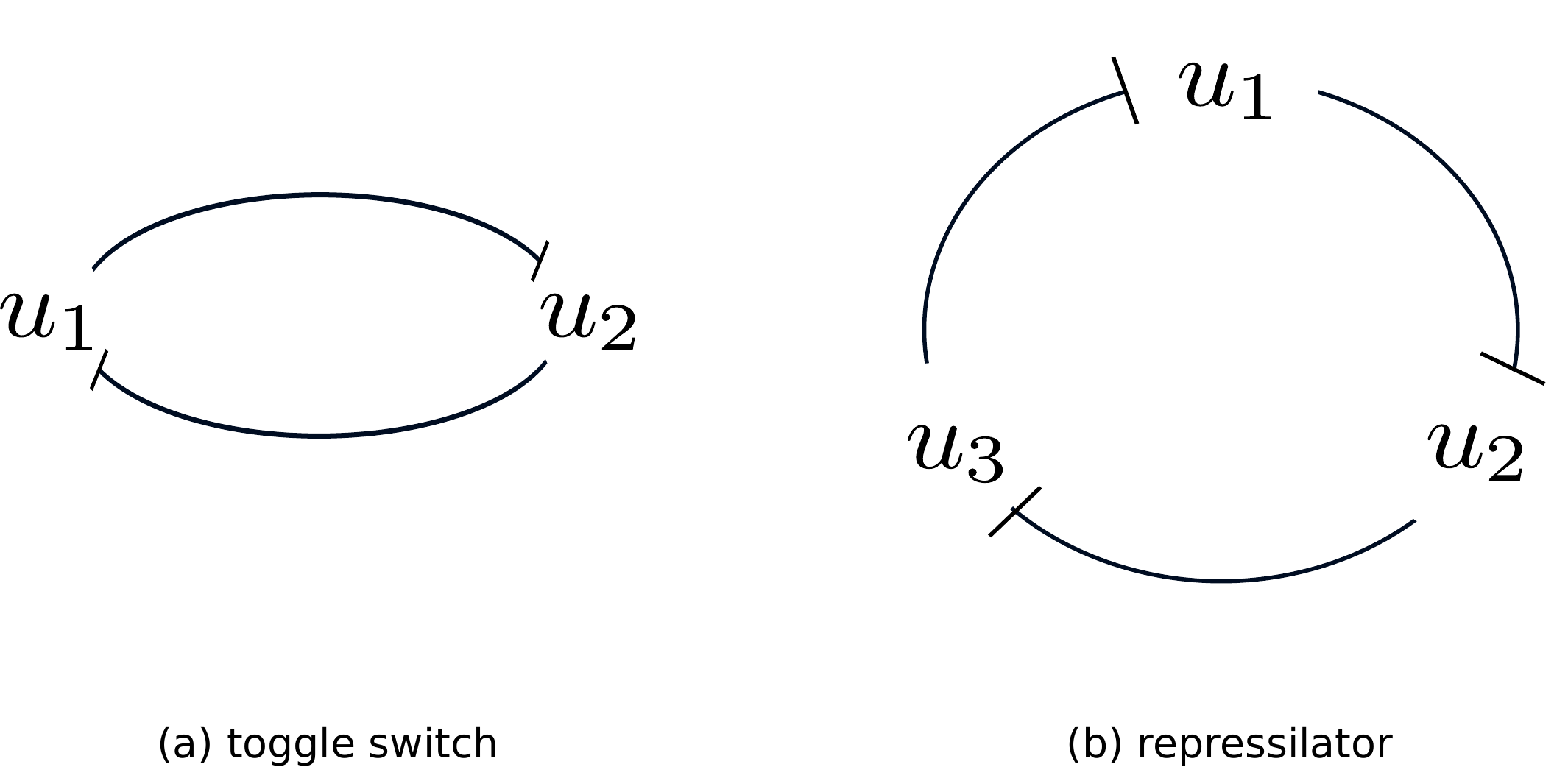}
\parbox{.9\textwidth}{
\caption{\footnotesize {(a) Nodes $u_1$, $u_2$ inhibiting each others activity resulting in a switch. The node which starts out stronger suppresses the activity of the other. (b) Nodes $u_1$, $u_2$, and $u_3$ suppress each other in a cyclic fashion. Under certain conditions, this can lead to oscillations.}%
}\label{fig:toggle_repressillator}}
\end{center}
\end{figure}

\subsection{A network of two mutually inhibiting elements}

We start with the common \emph{toggle switch} motif, \emph{i.e} a network of
two mutually repressing elements (see Fig.~\ref{fig:toggle_repressillator}a) ~[\cite{NovakPatakiCilibertoTyson2003,GardnerCantorCollins2000}]. Let
$(u_1,u_2) \in [0,1]^2$ represent the normalized levels of activity at the two nodes. 
Therefore, when $u_i = 1$ the $i^{\text{th}}$
network element is maximally active (expressed). The activity of the two nodes in the system 
can be modeled by
\begin{align}\label{EquationForToggleSwitch}
\begin{split}
\frac{du_1}{dt} &= 0.5\frac{1-u_1}{J+1-u_1}-u_2\frac{u_1}{J+u_1},  \\
\frac{du_2}{dt} &= 0.5\frac{1-u_2}{J+1-u_2}-u_1\frac{u_2}{J +u_2},
\end{split}
\end{align}
where $J$ is some positive constant. The structure of
Eq.~\eqref{EquationForToggleSwitch} implies that the cube $[0,1]^2 = \{(u_1, u_2)
\,|\, 0 \le u_1, u_2 \le 1 \}$ is invariant (see
Proposition~\ref{prop:invariantCube}).

\subsubsection{Piecewise linear approximation}

In the limit of small $J$, Eq.~\eqref{EquationForToggleSwitch}  can be
approximated by a piecewise linear differential equation:
If $u_i$ is not too close to zero the expression $u_i/(J+u_i)$ is approximately unity.
  More precisely, we fix a small $\d > 0 $, which will be chosen to depend on
$J$.  When $u_i > \d$ and $J$ is small then $u_i/(J+u_i) \approx1$. Similarly, when 
$u_i > 1-\d$ then  $(1-u_i)/(J+1-u_i) \approx 1$.

With this convention in mind we  break the cube $[0,1]^2$ into several
subdomains, and define
a different reduction of Eq.~\eqref{EquationForToggleSwitch} within each. Let $\mathcal{R}_S^T$ to denote the region where $S$ is the set of variables that
are close to 0, and $T$ is the set of variables close to 1 (See Table~\ref{table:2d} and Eq.~\eqref{Rdef}). Also, we omit the curly brackets and commas in $\mathcal{R}_S^T$ (e.g.  $\mathcal{R}^{\{\}}_{\{\}}=\mathcal{R}$ and $\mathcal{R}_{\{1\}}^{\{\}}=\mathcal{R}_1$) (see Fig.~\ref{fig:subdomains}). 

\begin{figure}[h]
\begin{center}
\includegraphics[scale = .65]{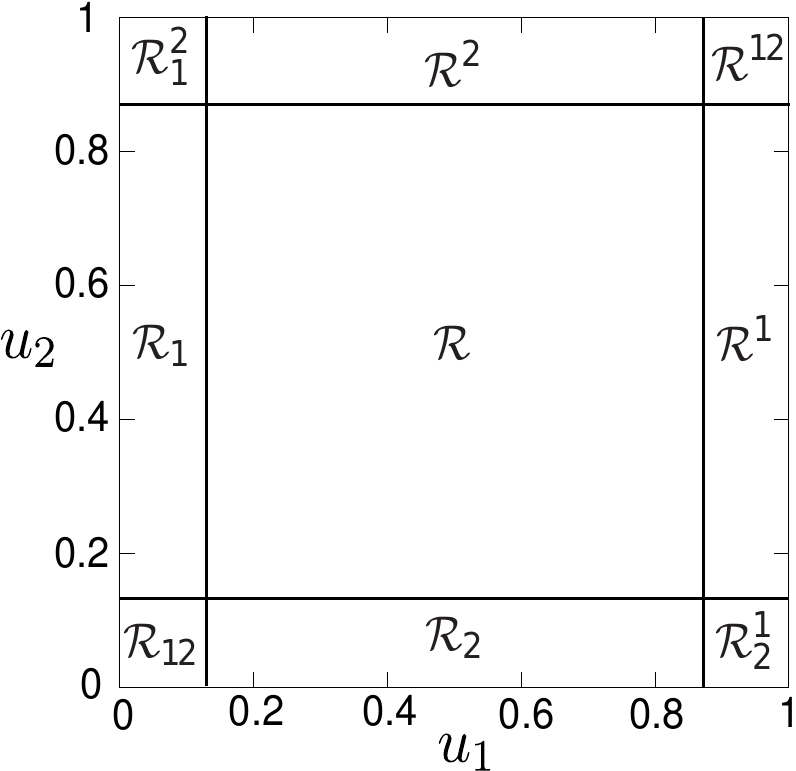}
\parbox{.9\textwidth}{
\caption{\footnotesize
Subdomains $\mathcal{R}_S^T$ for the unit square $[0,1]^2$. 
\label{fig:subdomains}}
}
\end{center}
\end{figure}

We first reduce Eq.~\eqref{EquationForToggleSwitch} on each of the subdomains.
The interior of the domain $[0,1]^2$ consist of points where neither coordinate is close to 0 nor 1, and
 is defined by
\begin{align}\label{R00}
\mathcal{R}:= \{(u_1,u_2) \in [0,1]^2 \,|\, \d \le u_1 \le 1-\d \text{ and }
\d \le u_2 \le 1-\d\}.
\end{align}
 Eq.~\eqref{EquationForToggleSwitch}, restricted to $\mathcal{R}$ is
approximated by the linear differential equation
\begin{align}\label{interiorLinear2d}
\frac{du_1}{dt} = 0.5-u_2, \quad
\frac{du_2}{dt} = 0.5-u_1.
\end{align}
 
On the other hand, if one of the coordinates is near the boundary, while the
other is in the interior, the approximation is different. For instance, the
region
\begin{align}\label{R20}
\mathcal{R}_2:= \{(u_1,u_2) \in [0,1]^2 \,|\,  u_2 < \d \text{ and } \d \le
u_1 \le 1-\d\},
\end{align}
forms a boundary layer where $u_2$ is of the same order as $J$. 
 
The term $u_2/(J + u_2)$ cannot be approximated by unity.  Instead the approximation takes the form
\begin{subequations}\label{edgeLinear2d}
\begin{align}
\frac{du_1}{dt} &= 0.5-u_2, \label{edgeLinear2da} \\
\frac{du_2}{dt} &= 0.5-u_1 \frac{u_2}{J + u_2}. \label{edgeLinear2db}
\end{align}
\end{subequations}
This equation can be simplified further.  Since $\mathcal{R}_2$ is invariant (for $u_1>.5$), $\frac{d u_2}{dt}$ must be small inside the boundary layer
$\mathcal{R}_2$ (see Fig.~\ref{fig:twoDimension_nullclines}).  We therefore use the approximations $u_2 \approx 0$ in 
Eq.~\eqref{edgeLinear2da} and $\frac{d u_2}{dt}
\approx 0$ in Eq.~\eqref{edgeLinear2db} to obtain
\begin{subequations}\label{smallu}
\begin{align}
\frac{du_1}{dt} &= 0.5, \label{smallua}\\
0 &= 0.5-u_1 \frac{u_2}{J + u_2}.\label{smallub}
\end{align}
\end{subequations}
Note that Eq.~\eqref{smallua} is linear  and decoupled from Eq.~\eqref{smallub},
while Eq.~\eqref{smallub} is an algebraic system which can be solved to obtain
$u_2 \approx J/(2u_1 - 1)$.
Within $\mathcal{R}_2$ we thus obtain the approximation 
\begin{subequations}\label{smallu2}
\begin{align}
u_1(t) & = 0.5 t  + u_1(0)  \label{E:smalleqa}\\
u_2(t) & = \frac{J}{ t + 2 u_1(0) - 1}  \label{E:smalleqb}
\end{align}
\end{subequations}

We only have the freedom of specifying
the initial condition $u_1(0)$, since  $u_2(0)$ is determined by the solution
of the algebraic equation~\eqref{smallub}.  As we explain below, this algebraic
equation defines a slow manifold within the subdomain $\mathcal{R}_2$.  The
reduction assumes that solutions are instantaneously attracted to this manifold. 

Table~\ref{table:2d} shows how this approach can be extended to all of $[0,1]^2$. 
There are  9  subdomains of the cube, one corresponding to the interior and four each to the 
edges and vertices.  On the latter eight subdomains, one or both variables are close to either 0 or 1.
Following the preceding arguments, variable(s)
close to 0 or 1 can be described by an algebraic equation. The resulting algebraic-differential systems are given in the last column of
Table~\ref{table:2d}. 
Furthermore, by using the approximations $u_i(t)\approx 0$ for $i\in S$ and $u_i(t)\approx 1$ for $i\in T$, we obtain a simple approximation of the dynamics in each subdomain which is $0$-th order in $J$. For example, in $\mathcal{R}_2$, we obtain the approximation $u_1(t) \approx 0.5t + u_1(0), u_2(t) \approx 0$.

Each approximate solution can potentially exit the subdomain within
which it is defined if at some time $u_i\approx 0$ or $u_i\approx 1$ and the $i$-th coordinate of the vector field is positive or negative, respectively. This can happen when the sign of some entry of the vector field changes; that is, solutions can exit subdomains when they reach a nullcline. Also, solutions can leave the subdomain if they started on the other side of the nullcline to begin with.
The global approximate solution of Eq.~\eqref{EquationForToggleSwitch} is 
obtained by using the exit point from one subdomain as the initial condition for
the approximation in the next.  In
subdomains other than $\mathcal{R}$ some of the initial conditions will be
prescribed by the algebraic part of the reduced system.  The global
approximation may therefore be discontinuous, as solutions entering a new
subdomain are assumed to instantaneously jump to the slow manifold defined by
the algebraic part of the reduced system.
 Fig.~\ref{fig:twoDimension} shows that when $J$ is small, this approach
provides a good approximation.

\begin{table}[t] 
\[
\begin{array}{c|c|c|rcl}
\hline
\text{Name of the subdomain}		&  u_1 	& 	u_2	&
\multicolumn{3}{c}{\text{Approximating system}}  \\
\hline
\multirow{2}{*}{$\mathcal{R}$}	&  \multirow{2}{*}{$ \d \le u_1 \le 1-\d
$}	&  \multirow{2}{*}{ $ \d \le u_2 \le 1-\d $}	&   	 u_1'	&=&
0.5-u_2,			\\
								&		
					&					
  		 &    u_2'		&=& 0.5-u_1			\\
\hline
\multirow{2}{*}{$\mathcal{R}^1$}	& \multirow{2}{*}{$ u_1 > 1-\d $} 	
& \multirow{2}{*}{$ \d \le u_2 \le 1-\d $}		&  	           	
  0	&=&  \displaystyle 0.5\frac{1-u_1}{J+1-u_1}-u_2,	\\
								&		
					&					
		 &  	 u_2'  &=& -0.5				       \\
\hline
\multirow{2}{*}{$ \mathcal{R}^2$}     & \multirow{2}{*}{$ \d \le u_1 \le 1-\d
$}		& \multirow{2}{*}{$ u_2 > 1-\d $}	&  u_1'	&=&	 -0.5,	
				\\
								&		
					&					
	  	 &                         0	&=&	\displaystyle
0.5\frac{1-u_2}{J+1-u_2}-u_1	\\
\hline
\multirow{2}{*}{$ \mathcal{R}_1 $}	& \multirow{2}{*}{$ u_1 < \d $} 
& \multirow{2}{*}{$ \d \le u_2 \le 1-\d $}    	& 			
0 	 &=& \displaystyle 0.5-u_2\frac{u_1}{J+u_1},  \\
								&		
					&					
	   	 & 		u_2'&=& 0.5  \\
\hline
\multirow{2}{*}{$ \mathcal{R}_2	$}	& \multirow{2}{*}{$ \d \le u_1
\le 1-\d $}		&   \multirow{2}{*}{$u_2 <  \d $}	& 	u_1'	
&=& 0.5 ,  \\
								&		
					&					
		 &			0		&=& \displaystyle 0.5
-u_1\frac{u_2}{J +u_2} \\
\hline
\multirow{2}{*}{$ \mathcal{R}^{12}$}& \multirow{2}{*}{$ u_1 > 1-\d$} 	& 
\multirow{2}{*}{$u_2 > 1-\d$} 	&	0	&=& \displaystyle
0.5\frac{1-u_1}{J+1-u_1}-1,	\\
								&		
					&					
		 &	0	&=& \displaystyle 0.5\frac{1-u_2}{J+1-u_2}-1
\\
\hline
\multirow{2}{*}{$ \mathcal{R}_{12}$}&  \multirow{2}{*}{$u_1 < \d$}	 & 
\multirow{2}{*}{$u_2 < \d$}	&	0	&=& \displaystyle 0.5- J
\frac{u_1}{J+u_1},		\\
								&		
					&					
		 &	0	&=& \displaystyle 0.5- J\frac{u_2}{J +u_2}	
\\
\hline
\multirow{2}{*}{$ \mathcal{R}_2^1	$}	& \multirow{2}{*}{$ u_1 > 1-\d$}
	& \multirow{2}{*}{$u_2 < \d $}		&	0	&=&
\displaystyle  0.5\frac{1-u_1}{J+1-u_1},	\\
								&		
					&					
		 &	0	&=& \displaystyle 0.5- \frac{u_2}{J +u_2}	
\\
\hline
\multirow{2}{*}{$ \mathcal{R}_1^2	$}	& \multirow{2}{*}{$ u_1 < \d $}
&\multirow{2}{*}{$  u_2 > 1-\d$}		&	0	&=&
\displaystyle 0.5-\frac{u_1}{J+u_1},		\\
								&		
					&					
		 &	0	&=& \displaystyle 0.5\frac{1-u_2}{J+1-u_2}
\\
\hline
\end{array}
\]
\begin{center}
\parbox{.8\textwidth}{
\caption{\footnotesize  List of differential--algebraic systems that approximate
Eq.~\eqref{EquationForToggleSwitch} in different parts of the domain. The
subdomains are named so that the superscript (subscript)  lists the coordinates
that are close to $1$ (close to 0), with 0 denoting the empty set.  For example,
$\mathcal{R}_1^2$ denotes that subdomain with $u_1 \approx 1$ and $u_2 \approx
0$, and $\mathcal{R}^2$ the subdomain where $u_2$ is near $1$, but $u_1$ is
away from the boundary. The middle column define the subdomain explicitly.   
The right column gives the differential-algebraic system that approximates
Eq.~\eqref{EquationForToggleSwitch}  within the given subdomain. }
\label{table:2d}
}
\end{center}
\end{table}
\begin{figure}[t]
\begin{center}
\includegraphics[scale = .45]{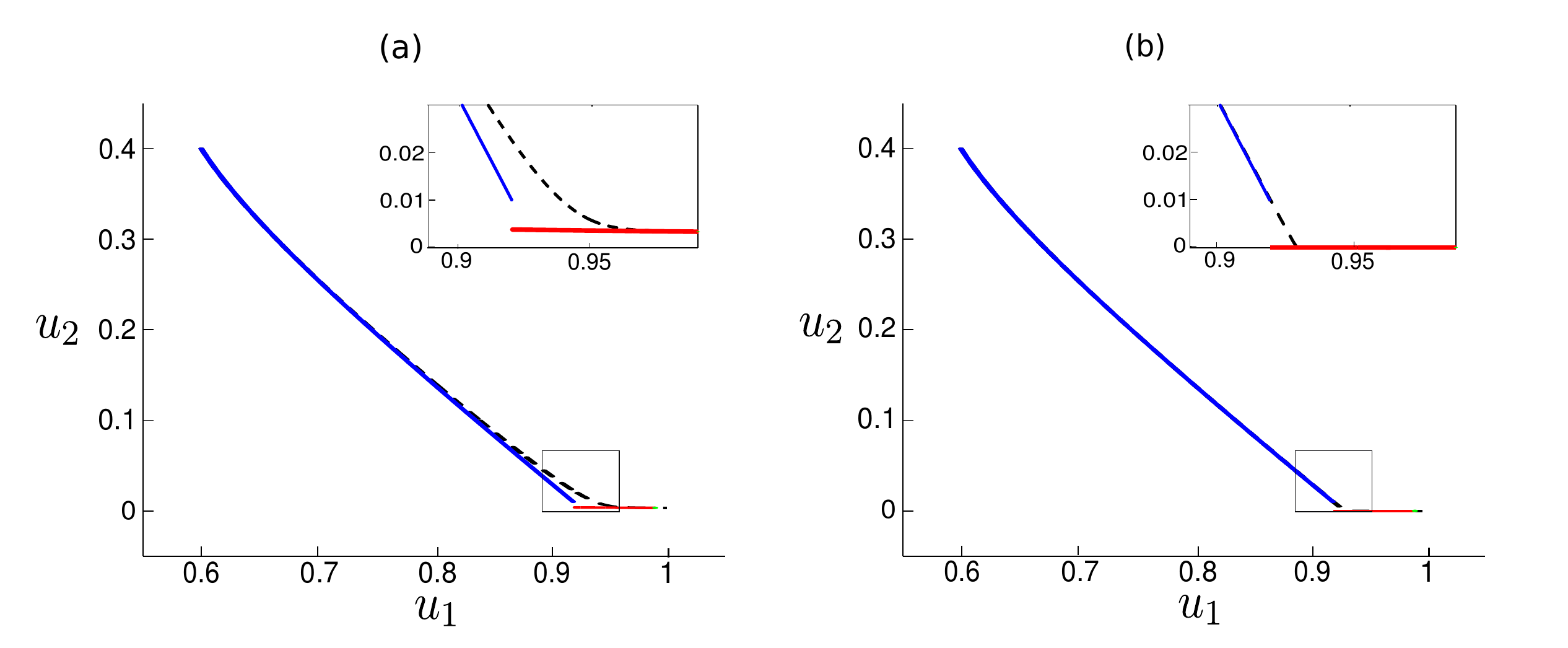}
\parbox{.9\textwidth}{
\caption{\footnotesize
Comparison of the numerical solution of Eq.~\eqref{EquationForToggleSwitch}
(dashed black) and the solution of the approximate system as listed in
Table~\ref{table:2d} (colors) for two different values of $J$. We used 
$J = 10^{-2}$ in (a); and $J =10^{-4} $ in (b).  Different colors are used for 
the solution of the reduced system in different subdomains.  Solution of the linear approximation started in the subdomain
$\mathcal{R}$ (Initial value: $u_1 = 0.6, u_2 = 0.4$), and as soon as $u_2$
decreased below $\d = 0.01$, we assumed that the solution entered subdomain  $\mathcal{R}_2$. 
The approximate solution is discontinuous since when $u_2 = \d$, the solution jumped
(see inset) to the manifold, described by the algebraic part of the linear
differential algebraic system prevalent in the subdomain $\mathcal{R}_2$, Eq.~\eqref{smallub}.
The solution finally stopped in the  subdomain $\mathcal{R}_2^1$.  }
\label{fig:twoDimension}
}
\end{center}
\end{figure}

\subsubsection{Boolean approximation}

We now derive a Boolean approximation, $h=(h_1,h_2):\{0,1\}^2\rightarrow\{0,1\}^2$, that captures certain qualitative features of  Eq.~\eqref{EquationForToggleSwitch}. The idea is to project small values of $u_i$ to 0 and  large values of $u_i$ to 1, and map  the value of the $i$-th variable into 0 and 1 depending on whether $u_i$ is decreasing or increasing, respectively. We will show that the resulting BN can be used directly to detect steady states in the corner subdomains. 

Note that for a BN time is discrete; a time step in the Boolean approximation can be interpreted as the time it takes the original system to transition between chambers. Different transitions in the Boolean network may have different duration in the original system; so the time steps in the BN are only used to keep track of the sequence of events, but not their duration.

The reduction described in the previous section gives a linear ODE  in the interior region $\mathcal{R}$ (Eq.~\eqref{interiorLinear2d}), where $\mathcal{R}$ approaches $[0,1]^2$ as $J\rightarrow 0$. The approximating  linear system therefore  provides  significant information about the behavior of the full, nonlinear system for $J$ small.

We first examine the nullclines. In Fig.~\ref{fig:twoDimension_nullclines} we can see that as $J$ decreases, in the interior of $[0,1]^2$ the nullclines of Eq.~\eqref{EquationForToggleSwitch} approach the nullclines of Eq.~\eqref{interiorLinear2d} given by $u_2=.5$ and $u_1=.5$ restricted to $[0,1]^2$.  These lines divide the domain into four chambers, which we denote
\begin{equation*}
\mathcal{C}_{12}:=[0,0.5)\times[0,0.5),\quad
\mathcal{C}_{1}^2:=[0,0.5)\times(0.5,1],\quad
\mathcal{C}_{2}^1 :=(0.5,1]\times[0,0.5),\quad
\mathcal{C}^{12}:=(0.5,1]\times(0.5,1].
\end{equation*}
\begin{figure}[h]
\begin{center}
\includegraphics[scale = .55]{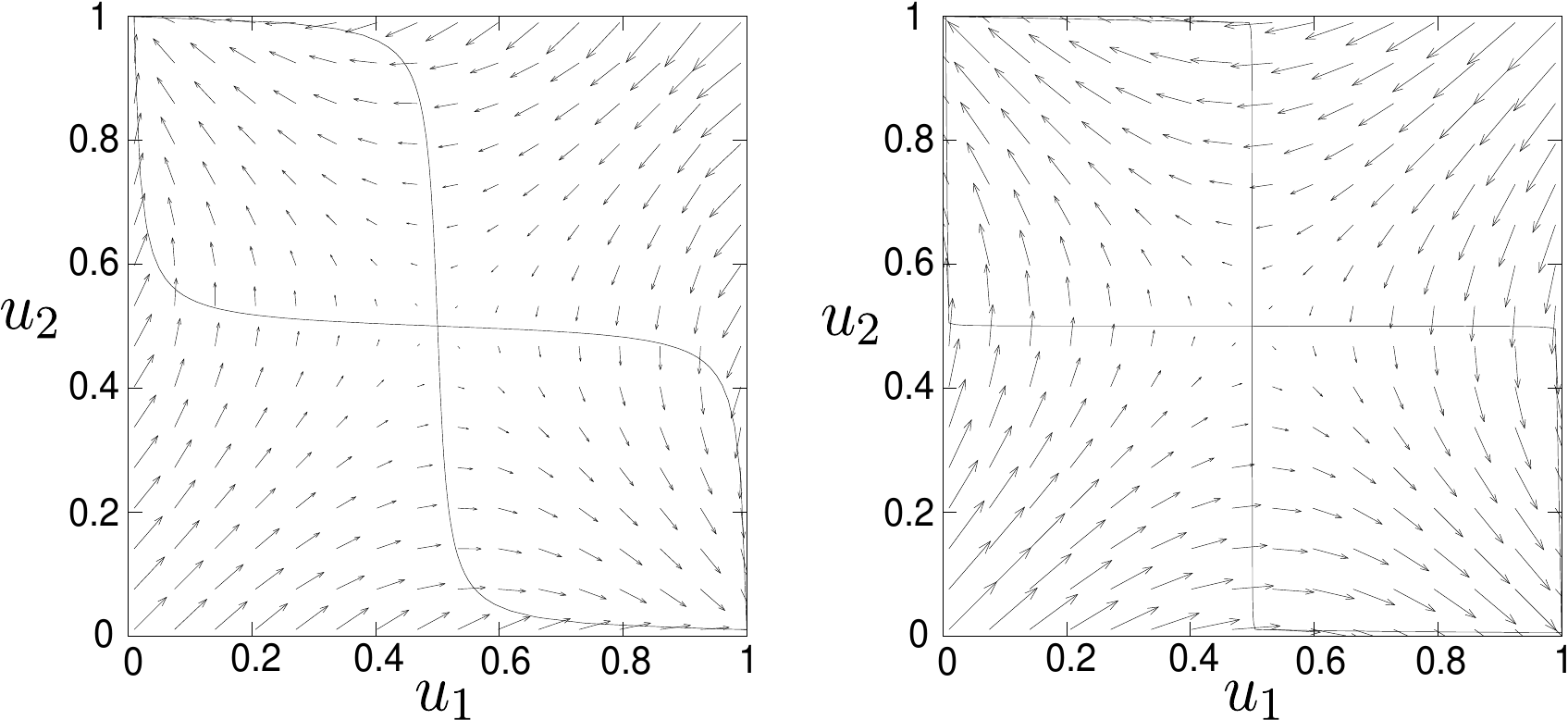}\\
\includegraphics[scale = .55]{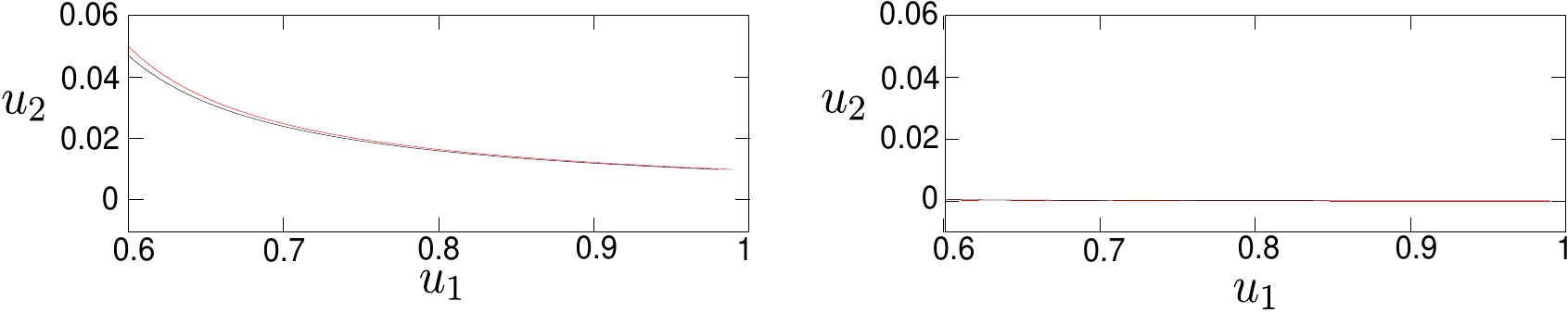}
\parbox{.9\textwidth}{
\caption{\footnotesize
Behavior of nullclines as $J$ decreases. Top: Nullclines of Eq.~\eqref{EquationForToggleSwitch} for $J=10^{-2}$ (left) and $J=10^{-4}$ (right). Bottom: Nullcline $\frac{du_2}{dt}=0$ of Eq.~\eqref{EquationForToggleSwitch} (black curve) and the manifold defined by Eq.~\eqref{smallub} (red) for $J=10^{-2}$ (left), and $J=10^{-4}$ (right). 
\label{fig:twoDimension_nullclines}}
}
\end{center}
\end{figure}

On the other hand, the part of the nullclines inside the boundary subdomains are approximately the slow manifolds defined by equivalents of Eq.~\eqref{E:smalleqb}. Here the slow manifolds converge to the nullclines  as $J \rightarrow 0$  (See Fig.~\ref{fig:twoDimension_nullclines}).

As a shorthand, we define the ``sign'' of a vector $v=(v_1\ldots,v_N)$ as the vector composed by the signs of its components, $sign(v):=(sign(v_i),\ldots,sign(v_N))$.
Note that although the sign of the vector $(.5-u_2,.5-u_1)$ is constant in each chamber, the sign of the vector field of Eq.~\eqref{EquationForToggleSwitch} may differ. For example, in chamber $\mathcal{C}_2^1$, the sign of the vector field can take all possible values. However, this difference is small when  $J$ is small, because the regions between the nullclines approach the actual chambers (Fig.~\ref{fig:twoDimension_nullclines}).  

We consider Eq.~\eqref{EquationForToggleSwitch} in each chamber, starting with the first coordinate, $u_1(t)$. 
For any solution with initial condition in $\mathcal{C}_{12}$, the sign of $u'_1(0)$ is positive and $u_1(t)$ increases within the chamber.  
We use this observation to define $h_1(\mathcal{C}_{12})=1$.  The formal definition of this function will be given below -- 
intuitively $h_i(\cdot)$ maps a chamber to 1 if $u_i$ is increasing within the chamber, and to 0 otherwise. 
Similarly, since $u_1(t)$ initially increases within $\mathcal{C}_2^1$, we let $h_1(\mathcal{C}_2^1)=1$. Similarly we  set $h_1(\mathcal{C}^{12})=0,$ $h_1(\mathcal{C}_1^2)=0,$ $h_2(\mathcal{C}_{12})=1$, $h_2(\mathcal{C}_1^2)=1$, $h_2(\mathcal{C}_2^1)=0$,  and $h_2(\mathcal{C}^{12})=0$.  The $i$-th variable is  ``discretized,'' \emph{i.e.} mapped to 0 and 1 depending on whether $u_i$ is decreasing or increasing, respectively. 

More formally, consider the set $\{0,1\}^2$, with each element identified with a chamber (e.g., the element $(0,1)$ represents the chamber $\mathcal{C}_1^2$). Then $h_1$ and $h_2$ are defined as Boolean functions from $\{0,1\}^2$ to $\{0,1\}$  by setting $h_1(0,0)=0, h_1(0,1)=0$, $h_1(1,0)=1$, $h_1(1,1)=0$, and $h_2(0,0)=0, h_2(0,1)=1$, $h_2(1,0)=0$, $h_2(1,1)=0$. These two Boolean functions define a BN, $h=(h_1,h_2):\{0,1\}^2\rightarrow\{0,1\}^2$.  The functions also define a dynamical system, $x(t+1)=h(x(t)), x \in \{0,1\}^2$. However, other update schedules   can be used [\cite{AracenaGolesMoreiraSalinas2009}].

The BN reduction can be obtained easily  from the sign of $(.5-u_2,.5-u_1)$ at the vertices of $[0,1]^2$,  since the sign of the vector field is constant within a chamber. To do so we use the the Heaviside  function, $H$, defined by $H(y)=0$ if $y<0$, $H(y)=1$ if $y>0$, and $H(0)=\frac{1}{2}$. For example, in $\mathcal{C}_{12}$, both entries increase. We  can see this by evaluating  $H(.5-u_2)=H(.5-u_1)=1$ for $u=(0,0)$.  Using the same argument in each chamber, we obtain the BN 
\begin{equation}\label{EquationForToggleSwitchBN}
h(x)=H(.5-x_2,.5-x_1),
\end{equation}
where we used the convention that $H$ acts entrywise on each component in the argument.

\subsubsection{Steady states of the BN and the PL approximation}

While the BN gives information about which variables increase and decrease within a chamber, it is not yet
clear how or if the dynamics of the BN in Eq.~\eqref{EquationForToggleSwitchBN}, and the PL approximation in Table \ref{table:2d} are related. 

We next show that the steady states of the PL approximation near the vertices can be determined by the steady states of the BN.  The reduced equations in the corner subdomains  $\mathcal{R}_{12},\mathcal{R}^{12}, \mathcal{R}_1^2,$ and $\mathcal{R}_2^1$ are purely algebraic. When $J$ is small, some of these equations have a solution in $[0,1]^2$, indicating a stable fixed point near the corresponding corner (in this case $\mathcal{R}_1^2$ and $\mathcal{R}_2^1$). Others will not have a solution in $[0,1]^2$, indicating that approximate solutions do not enter the corresponding subdomain (here $\mathcal{R}_{12}$ and $\mathcal{R}^{12}$). To make the relationship between steady states less dependent on the actual parameters, consider the system
\begin{align}
\begin{split}\notag
\frac{du_1}{dt} &= b_1^+\frac{1-u_1}{J+1-u_1}-(u_2+b_1^-)\frac{u_1}{J+u_1},  \\
\frac{du_2}{dt} &= b_2^+\frac{1-u_2}{J+1-u_2}-(u_1+b_2^-)\frac{u_2}{J +u_2},
\end{split}
\end{align}
where $x^+=\max{\{x,0\}}$ and $x^-=\max\{-x,0\}$. In the previous example $b_1=b_2=.5$, $b_1^+=b_2^+=.5$ and $b_1^-=b_2^-=0$.  

Now, at the corner subdomain $\mathcal{R}_{12}$ we have the approximate equations,
\[
0= b_1^+ -b_1^-\frac{u_1}{J+u_1},  0= b_2^+ -b_2^-\frac{u_2}{J +u_2},\]
or equivalently,
\[u_1 =\frac{ -b_1^+ J}{b_1}, u_2 =\frac{ -b_2^+ J}{b_2}.\]

These equations have a solution in $\mathcal{R}_{12}$ if and only if $b_1<0$ and $b_2<0$, or equivalently, if and only if $H(b_1,b_2)=(0,0)$. A similar analysis  leads to the following conditions for the existence of fixed points in each corner subdomain 
\begin{align}
\begin{split}\notag
&\textrm{On $\mathcal{R}_{12}$ :}\ \ H(b_1,b_2)=(0,0),\qquad
\textrm{On $\mathcal{R}^{12}$ :}\ \ H(b_1-1,b_2-1)=(1,1),\\
&\textrm{On $\mathcal{R}^2_1$ :}\ \  H(b_1-1,b_2)=(0,1),\qquad
\textrm{On $\mathcal{R}^1_2$ :}\ \  H(b_1,b_2-1)=(1,0).
\end{split}
\end{align}
More compactly, the condition is $H(b_1-x_2,b_2-x_1)=(x_1,x_2)$, where $x=(x_1,x_2)$ is the corner of interest. 
Hence, the BN can also be used directly to detect which corner subdomains contain steady states. 

The relationship between steady states in the full system at the corner subdomains and the steady states of the BN is straightforward. However, since there are many update schemes for BNs, the relationship between trajectories is more subtle. For example, using  synchronous update we obtain the transition $(0,0)\rightarrow (1,1)$ which is not compatible with the solutions of Eq.~\eqref{EquationForToggleSwitch} (See Fig.~\ref{fig:twoDimension_ODEandBN}). On the other hand, using  asynchronous update we obtain the transitions $(0,0)\rightarrow (1,0)$ and $(0,0)\rightarrow (0,1),$ which are more representative of the solutions of Eq.~\eqref{EquationForToggleSwitch}. Thus, we will focus on transitions that are independent of the update scheme, that is, transitions where only one entry changes.

\begin{figure}[h]
\begin{center}
\includegraphics[scale = .6]{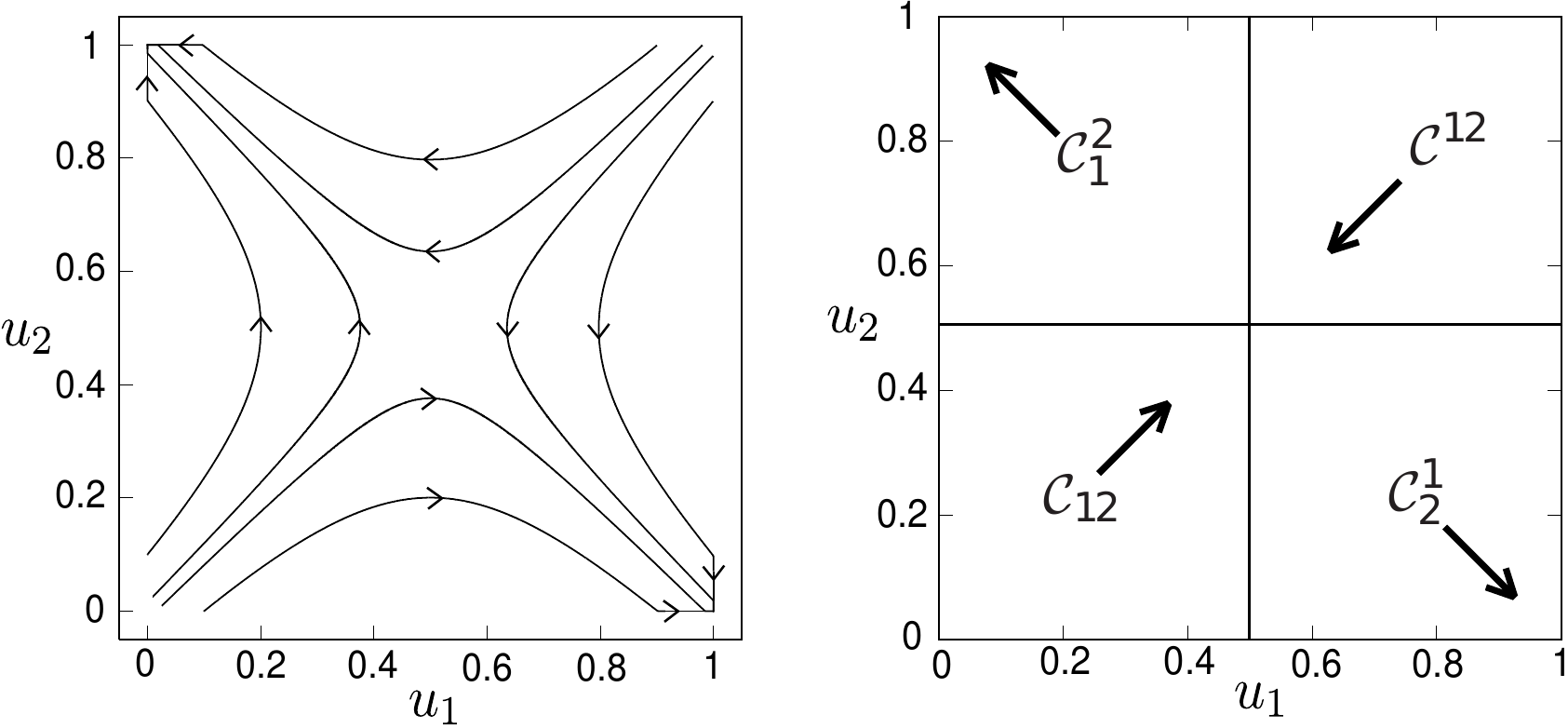}
\parbox{.9\textwidth}{
\caption{\footnotesize
Left: Solutions of Eq.~\eqref{EquationForToggleSwitch} for $J=10^{-4}$. When a solution is close to the boundary regions of $\mathcal{C}_1^2$ and $\mathcal{C}_2^1$, they enter the invariant region as shown in Fig.~\ref{fig:twoDimension}b. Right: Graphical representation of the Boolean transitions ($00\rightarrow 11$, $11\rightarrow 00$, $01\rightarrow 01$, $10\rightarrow 10$).  
\label{fig:twoDimension_ODEandBN}
}} 
\end{center}
\end{figure}

\subsection{A network of three mutually inhibiting elements}\label{sec:3nodes}

The same reduction can be applied to systems of arbitrary dimension.  As an
example consider
the
\emph{repressilator}~[\cite{NovakPatakiCilibertoTyson2003,ElowitzLeibler2000}] (see Fig.~\ref{fig:toggle_repressillator}a)
described by
\begin{eqnarray}\label{threeNodeMM}
 \frac{du_1}{dt} &=& 0.6\frac{1- u_1}{J+1-u_1} -  u_3\frac{ u_1}{J+u_1}, \notag
\\
 \frac{du_2}{dt} &=& 0.4\frac{1- u_2}{J+1-u_2} -  u_1\frac{ u_2}{J+u_2}, \\
 \frac{du_3}{dt} &=& 0.3\frac{1- u_3}{J+1-u_3} -  u_2\frac{ u_3}{J+u_3}. \notag
\end{eqnarray}
The cyclic repression of the three elements in this network leads to oscillatory
solutions over a large range of values of $J$.  The domain of this system,
$[0,1]^3$, can
be divided into 27 subdomains corresponding to 1 interior,
6 faces, 12 edges, and 8 vertices.  

We can again  approximate Eq.~\eqref{threeNodeMM} with solvable
differential--algebraic equation within each  subdomain, to obtain a global
approximate solution (See  Fig.~\ref{fig:threeDimension}). Note that both the numerically obtained solution to  Eq.~\eqref{threeNodeMM} and the solution to the piecewise linear equation exhibit oscillations, and that the approximation is discontinuous. Again, in the limit $J \rightarrow 0$ we obtain a continuous 0-th order approximation.

In this singular limit, solutions can exit a subdomain when they reach a nullcline of the linear system. For example, when $u_2$ is close to 0 and a solution transitions from $u_1>.4$ to $u_1<.4$, the sign of the second entry of $(0.6-u_3,0.4-u_1,0.3-u_2)$ changes from negative to negative; so the second entry of the solution starts increasing (see Fig.~\ref{fig:threeDimension}, panel (e)).  Solutions therefore leave the subdomain on which $u_2\sim  J$ is small and enter the subdomain where $u_2\gg J$. Similarly when $u_1$ is close to 1 solutions transitions from $u_3<.6$ to $u_3>.6$, and the sign of the first entry of $(0.6-u_3,0.4-u_1,0.3-u_2)$ changes from positive to negative.  Hence the first entry of the solution starts decreasing (see Fig.~\ref{fig:threeDimension}, panel (f)), and solutions leaves the subdomain where $1-u_1\sim  J$  and enter another where $1-u_1\gg J$.

The BN corresponding to Eq.~\eqref{threeNodeMM},  $h=(h_1,h_2,h_3):\{0,1\}^3\rightarrow\{0,1\}^3,$ is given by $h(x)=H(0.6-x_3,0.4-x_1,0.3-x_2)$, where $H$ is the Heaviside function acting entry wise on the arguments.
Eq.~\eqref{threeNodeMM} does not have steady states at the corner subdomains, and neither does the corresponding  BN. A subset of states belong to a periodic orbit of the BN: 
$$ (0,1,1) \rightarrow (0,1,0)\rightarrow(1,1,0)\rightarrow(1,0,0)\rightarrow(1,0,1)\rightarrow(0,0,1).
$$  
Note that subsequent states in this orbit differ in a single entry. Thus, the transitions between the states have an unambiguous interpretation in the original system: The BN predicts that if the initial condition is in chamber $\mathcal{C}_{12}^3$, then solutions of  Eq.~\eqref{threeNodeMM}  will go to chamber $\mathcal{C}_{1}^{23}$, then to $\mathcal{C}_{13}^2$, and so on. Indeed, solutions of  Eq.~\eqref{threeNodeMM}, are attracted to a periodic orbit that transitions between the chambers in this order.  The remaining two states form a period two orbit under synchronous update, $(1,1,1) \leftrightarrow (0,0,0)$.   Here the BN does not give precise information about the dynamics of the original system. We will show that under certain conditions, orbits of the BN where only entry changes 
at each timestep, correspond to oscillations in the original system.

\begin{figure}[h]
\centering
 \includegraphics[width = .95\textwidth]{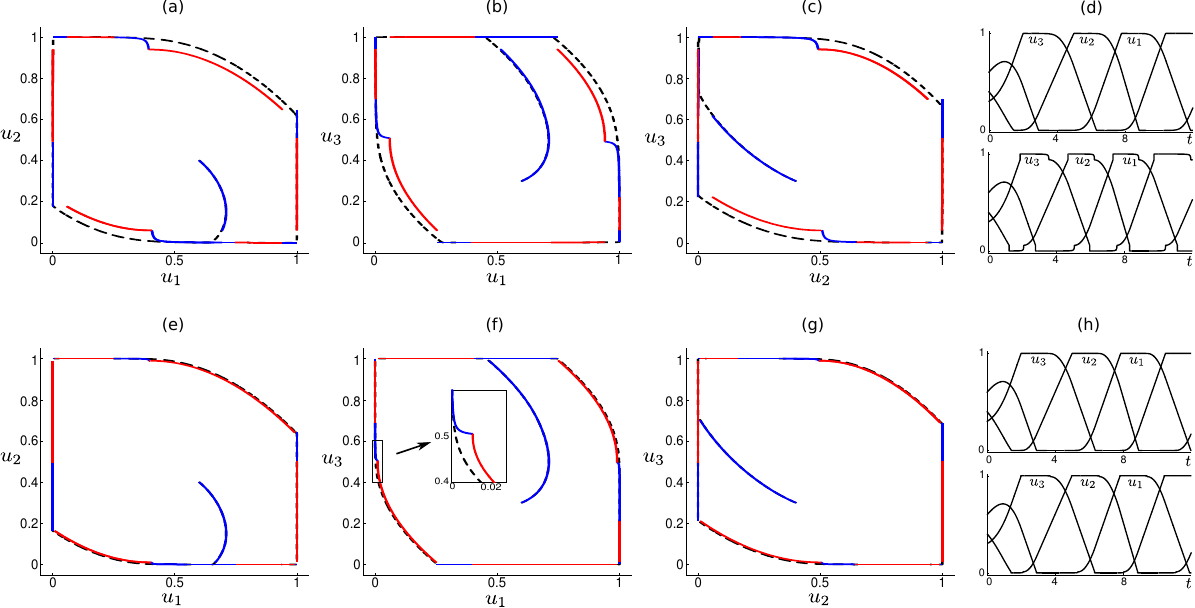}
\parbox{.9\textwidth}{
\caption{\footnotesize
Comparison of the numerical solution of Eq.~\eqref{threeNodeMM} (dashed black)
and the solution of the approximate piecewise linear system (colors) for
two different sets of $J$ and $\d$. For (a)-(c) $J = 10^{-2}, \d = 0.06$; for
(e)-(g)  $J = 10^{-4}, \d = 0.01$.  The approximate solution changes color when
switching between different subdomains.  
Panel (d) shows the time series for the solutions of Eq.~\eqref{threeNodeMM} (top) and the PL system (bottom), corresponding to (a)-(c). Panel (h) shows the time series for the solutions of Eq.~\eqref{threeNodeMM} (top) and the PL system (bottom), corresponding to (e)-(g). \label{fig:threeDimension}
}}
\end{figure}

\section{General reduction  of the model system}\label{GeneralTheory}

The approximations described in the previous section can be extended to the more general model given in Eq.~\eqref{eqn:ProblemEquation}: 
\begin{align}\notag
	\frac{du_i}{dt} = A_i
\frac{1-u_i}{J_{i}^A+1-u_i}-I_i\frac{u_i}{J_{i}^I+u_i},
\end{align}
where $J_i^A,J_i^I$ are some positive constants.  Here $A_i$ and $I_i$ are
activation/inhibition functions that capture the impact of other variables on
the evolution of $u_i$.   The initial conditions are assumed to satisfy $u_i(0)
\in [0,1]$ for all $i$.

 We assume that the activation and inhibition functions are both
affine~[\cite{NovakPatakiCilibertoTyson2001,de02}],
 \begin{equation} \label{activation/Inhibition}
  A_i := \sum_{j=1}^N{w_{ij}^+u_j} +b_i^+,
\quad
   I_i := \sum_{j=1}^N{w_{ij}^-u_j} +b_i^-,
 \end{equation}	
 where  we use the convention $x^+ = \text{max}\{x,0\}$ and $x^- =
\text{max}\{-x,0\}$.
The  $N\times N$ matrix, $W = [w_{ij}]$ and the  $N \times 1$ vector  $b = [\
b_1 \  b_2 \  ...\  b_N\ ]^t$ capture the connectivity and external input to the
network, respectively. In particular,  $w_{ij}$
gives  the contribution  of the $j^{\text{th}}$ variable to the growth rate of
$i^{\text{th}}$ variable. If $w_{ij} > 0 $, then $w_{ij}$  appears in the
activation function for $u_i$; and if $w_{ij}<0$ then $-w_{ij}$  appears in the
inhibition function for $u_i$.
The intensity of the external input to the $i^{\text{th}}$ element is $|b_i|$,
and it contributes to the activation or the inhibition function, depending on
whether $b_i > 0$ or $b_i < 0$,  respectively. 

\begin{proposition}\label{prop:invariantCube}
The cube $[0,1]^N$ is invariant for Eq.~\eqref{eqn:ProblemEquation}.
\end{proposition}
\begin{proof}
It will be enough to show that the vector field at any point on the boundary is not
directed outward. Since, $A_i \geq 0$ and $I_i \geq 0$, for any $i$,
\begin{align*}
	\frac{du_i}{dt}\bigg|_{u_i = 0} = A_i\frac{1}{J_{i}^A+1} \ge 0,
\quad
\text{and}
\quad
	\frac{du_i}{dt}\bigg|_{u_i = 1} = -I_i\frac{1}{J_{i}^I+1} \le 0.
\end{align*}
\end{proof}

\subsection{The PL approximation}
To obtain a solvable reduction of Eq.~\eqref{eqn:ProblemEquation} we follow the 
procedure outlined in  Section~\ref{ExampleProblems}.  We first present the results, and provide the mathematical justification in the next section.  For
notational convenience we let $J_i^A = J_i^I = J$, with  $0 < J \ll 1$.   
The general case is equivalent.  We will use  $\d = \d(J) > 0 $  to define the thickness of the boundary layers.  We start with the subdivision of the
$N$-dimensional cube, $[0,1]^N$.

Let $T$ and $S$ be two disjoint subsets of  $\{1,2,...,N\}$, and let
\begin{align} \label{Rdef} 
\mathcal{R}^{T}_{S}:=
\Big\{
(u_1,u_2,...,u_N) \in [0,1]^N  \,\Big|\,
 u_{s} 		 <  \d \text{ for all }   s \in S; \quad
    &u_{t}	 	> 1-\d\text{ for all }   t \in T; \notag \\
\text{and }
 \d \le  &u_k\le 1-\d
  \text{ for all }  \quad k  \notin S \cup T
\Big \}.
\end{align}
We extend the convention used in Table~\ref{table:2d}, and in Eqs.~(\ref{R00})
and (\ref{R20}) so that  $\mathcal{R}^{T} :=\mathcal{R}^{T}_{S}$ when $S$ is
empty;  $\mathcal{R}_S :=\mathcal{R}^{T}_{S}$ when $T$ is empty; and
 $\mathcal{R} :=\mathcal{R}^{T}_{S}$  when $T$, $S$ are both  empty.

Within each subdomain $\mathcal{R}_S^T$,  Eq.~\eqref{eqn:ProblemEquation} can be
approximated by a
 different linear differential--algebraic system. Following the reduction from
Eq.~\eqref{EquationForToggleSwitch} to Eq.~\eqref{edgeLinear2d}, for $i \notin S
\cup T$ we obtain the linear system
\begin{subequations}\label{equationInRST}
\begin{equation} 
 \frac{du_i}{dt} = \sum_{j = 1}^N w_{ij}u_j +  b_i. \label{four}
\end{equation}
For $s \in S$ one of the nonlinear terms remains and we  obtain
\begin{equation} 
\frac{du_s}{dt} = \left(\sum_{j = 1}^N 
w_{sj}^+u_j +  b_s^+\right)-\left(\sum_{j= 1}^N 
w_{sj}^-u_j +  b_s^-\right)\frac{u_s}{J+u_s}, \label{five}
\end{equation}
while for $t \in T$ we will have
\begin{equation} 
\frac{d u_t}{dt} = \left(\sum_{j = 1}^N
w_{tj}^+u_j + b_t^+\right)\frac{1-u_t}{J+1-u_t}
-\left(\sum_{j = 1}^N w_{tj}^-u_j +  b_t^-\right). \label{six}
\end{equation}
\end{subequations}
Eq.~\eqref{equationInRST} is simpler than Eq.~\eqref{eqn:ProblemEquation}, but it is
not solvable yet.  Following the reduction from Eq.~\eqref{edgeLinear2d} to
Eq.~\eqref{smallu},  we now further reduce Eqs.(\ref{five}--\ref{six}).
First we use the approximations $u_s \approx 0$ and $u_t \approx 1$ in the
activation and inhibition functions appearing in Eq.~\eqref{equationInRST}.
Second, we assume
that $u_s$ for $s \in S$ and $u_t$ for $t \in T$ are in steady state.

Under these assumptions we obtain the reduction of Eq.~\eqref{eqn:ProblemEquation}
within any subdomain $\mathcal{R}_S^T$
\begin{subequations}\label{eqn:mainReduced_RST}
\begin{align}
\frac{du_i}{dt} &=  \sum_{j  \notin S \cup T}w_{ij}u_j + \sum_{j \in T}w_{ij} 
+b_i
& i \notin S \cup T;  \label{eqn:mainReduced_RSTa} \\
0 &=  \sum_{j \notin S \cup T}w_{sj}^+u_j + \sum_{t \in T}w_{st}^+
  +b_s^+ -\left( \sum_{j \notin S \cup T }w_{sj}^-u_j + \sum_{t \in T}w_{st}^- +
b_s^-\right)\frac{u_s}{J+u_s},
& s \in S; \label{eqn:mainReduced_RSTb} \\
0 &= -\left( \sum_{j \notin S \cup T }w_{tj}^+u_j + \sum_{j \in
T}w_{tj}^++b_t^+\right)\frac{1-u_t}{J+1-u_t}
+\sum_{j \notin S \cup T}w_{tj}^-u_j + \sum_{j \in T}w_{tj}^- + b_t^-,
& t \in T. \label{eqn:mainReduced_RSTc}
\end{align}
\end{subequations}

 Eq.~(\ref{eqn:mainReduced_RST}) is  solvable since Eq.~(\ref{eqn:mainReduced_RSTa}) is
decoupled, and   Eqs.(\ref{eqn:mainReduced_RSTb}) and
(\ref{eqn:mainReduced_RSTc}) are solvable for ${u}_s$ and ${u}_t$, respectively, as
functions of the solution of Eq.~(\ref{eqn:mainReduced_RSTa}).

Note that in the singular limit $J=0$ we obtain the $0$-th order approximations:
\begin{subequations}\notag
\begin{align}
\frac{du_i}{dt} &= \sum_{j\notin S \cup T}w_{ij}u_j + \sum_{j \in T}w_{ij}+b_i
& i \notin S \cup T;  \notag \\
u_s &=  0, & s \in S; \notag \\
u_t &= 1, & t \in T. \notag
\end{align}
\end{subequations}

\subsection{Boolean approximation}

To obtain the Boolean approximation we follow the process described in Section \ref{ExampleProblems}.  We consider the chambers determined by the complement of the union of the $N$ hyperplanes $\sum_{j=1}^N w_{ij} u_j +b_i= 0$ (restricted to $[0,1]^N$)  where $i=1,\ldots, N$. We denote with $\Omega$ the set of all chambers $\Omega:=\{\mathcal{C}: \mathcal{C} \textrm{ is a chamber}\}$; alternatively, $\Omega$ is the set of connected components of $[0,1]^N\setminus \cup_{i=1}^N\{u:\sum_{j=1}^N w_{ij} u_j +b_i=0\}$. We assume that $\sum_{j=1}^N w_{ij} x_j +b_i\neq 0$ for all $i=1,\ldots, N$ and for all $x\in\{0,1\}^N$. This guarantees that each corner of $[0,1]^N$ belongs to a chamber. The set of parameters excluded by this assumption has measure zero.

Let $S$ and $T$ be two disjoint subsets of $\{1, 2, \ldots, N\}$ such that $S\cup T=\{1, 2, \ldots, N\}$ and let $x\in\{0,1\}^N$ be the corner that belongs to the corner subdomain $\mathcal{R}_S^T$. Note that $x_i=0$ for $i\in S$ and $x_i=1$ for $i\in T$. The chamber $\mathcal{C} \in \Omega$ that contains the corner in subdomain $\mathcal{R}_S^T$
will be denoted by  $\mathcal{C}_S^T$.  We do not name the remaining chambers.

In general, the chambers can be more complex than in the examples of  Section \ref{ExampleProblems}. Chambers do not have to be hypercubes,  different corners  may belong to the same chamber, and some chambers may not even contain a corner of $[0,1]^N$, as illustrated in Fig.~\ref{fig:extra_examples_hyperplanes}.  In the first example, $(0,1)$ and $(1,1)$ belong to the same chamber, that is, $\mathcal{C}_1^2=\mathcal{C}^{12}$, and neither  $\mathcal{C}_{12}$ containing $(0,0),$ nor $\mathcal{C}_2^1$ containing $(1,0)$ are rectangles. Also, $\Omega$ has three elements: $\mathcal{C}_{12}$, $\mathcal{C}_2^1$, and $\mathcal{C}_1^2=\mathcal{C}^{12}$. In the second example, two chambers do not contain any corner of $[0,1]^2$ and are not named. Hence, $\Omega$ has four elements: $\mathcal{C}_{12}=\mathcal{C}_1^2$, $\mathcal{C}_2^1=\mathcal{C}^{12}$, and two unnamed chambers that contain no corners.

\begin{figure}[h]
\begin{center}
\includegraphics[scale = .5]{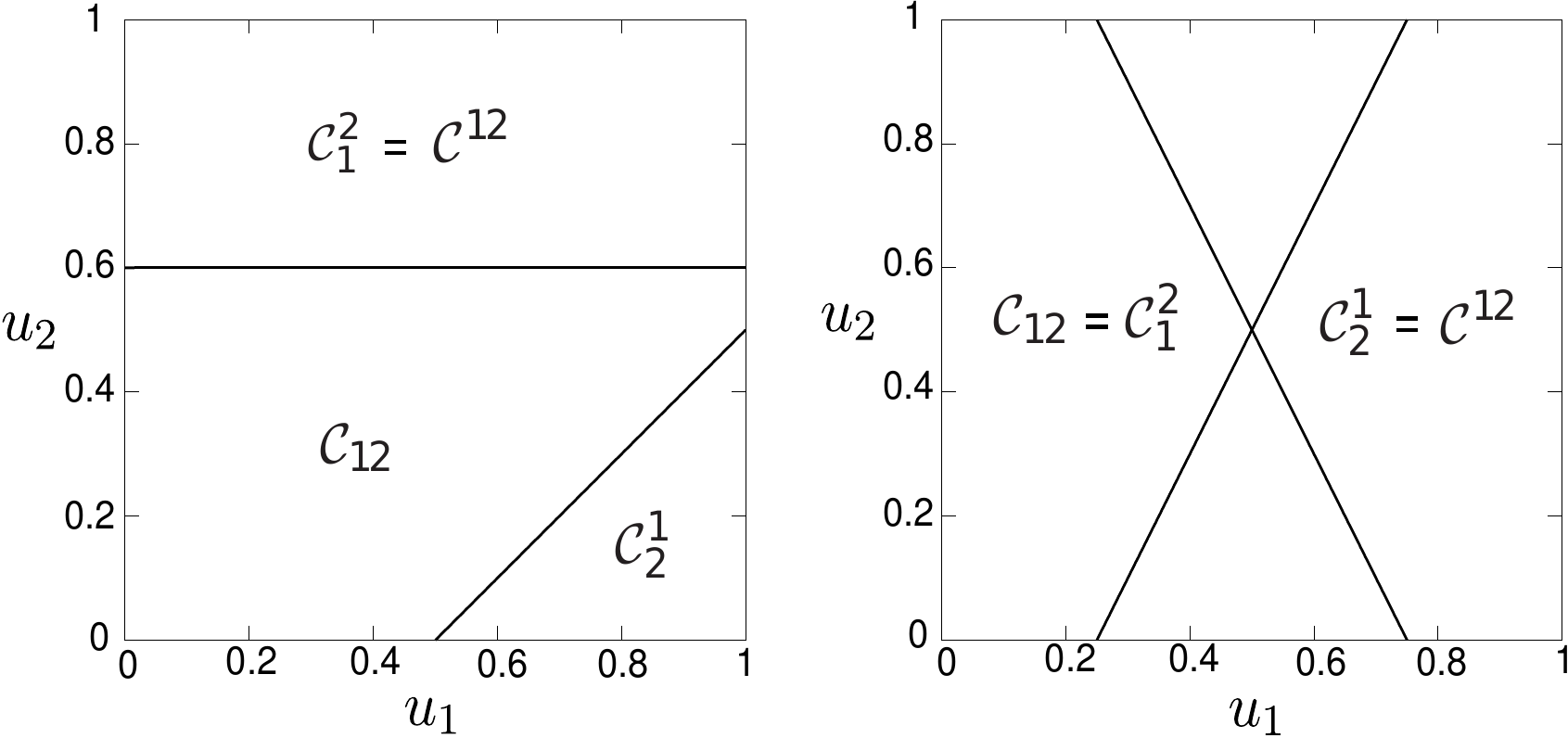}
\parbox{.9\textwidth}{
\caption[fig]{ \footnotesize
Chambers for $W=\left[\begin{array}{rr} -1 & -1 \\ 0 & -1 \end{array}\right]$ and 
$b=\left[\begin{array}{r} -0.5 \\ -0.6 \end{array}\right]$ (left);  
$W=\left[\begin{array}{rr} 2 & -1 \\ 2 & 1 \end{array}\right]$ and
$b=\left[\begin{array}{r} -0.5 \\ -1.5 \end{array}\right]$ (right).  
\label{fig:extra_examples_hyperplanes}
}} 
\end{center}
\end{figure}

To define the BN, $h=(h_1,\ldots,h_N):\{0,1\}^N\rightarrow \{0,1\}^N$ at $x\in\{0,1\}^N$, we need to find the signs of the components of the vector field $Wu+b$ on the chamber that contains $x$. Consider  $x \in \mathcal{R}_S^T$. Within $\mathcal{C}_S^T$ the signs of the components of $Wu+b$ do not change and are equal to the signs of the components of $Wx+b$. If the sign of the $i$-th component is negative, we let $h_i(x)=0,$ and if the sign is positive we let $h_i(x)=1$. In general, we can write 
\begin{equation}\label{eqn:mainReduced_BN}
h_i(x) =  H\left(\sum_{j=1}^N w_{ij}x_j+b_i\right),   \qquad \textrm{ or in vector form}  \qquad
h(x)=H(Wx+b).
\end{equation}
Hence the value of the BN at a corner $x\in\mathcal{R}_S^T$ is given by the  Heaviside function, applied entrywise  to $Wu+b$. Note that corners that are in the same chamber get mapped to the same point.

Importantly, using Eq.~\eqref{eqn:mainReduced_BN} we can compute the BN directly from Eq.~\eqref{eqn:ProblemEquation}. For example, for Eq.~\eqref{EquationForToggleSwitch} we have  $h(x_1,x_2)=H(0.5-x_2,0.5-x_1)$; and for  Eq.~\eqref{threeNodeMM} we have $h(x_1,x_2,x_3)=H(0.6-x_3,0.4-x_1,0.3-x_2)$.

Below we show that up to a set of small measure, the BN preserves information about the steady states of the original system.  We will also show that under some conditions, ``regular'' trajectories of a BN correspond to trajectories in the original system.

\section{Mathematical justification}\label{sec:math_justification}

 We next justify the different approximations made above: In Section \ref{sec:math_PL} we use Geometric Singular Perturbation Theory (GSPT) to justify the PL approximation. In Section \ref{sec:math_BN} we show that steady state information is preserved by the BN and that, under certain conditions, the BN also provides qualitative information about the global dynamics of the original system.

\subsection{Piecewise linear approximation}\label{sec:math_PL}

To obtain the reduced equations at the boundary of $[0,1]^N$, we define the following rescaled variables 
\begin{equation} \label{seven}
 \tilde{u}_s := \frac{u_s}{J}  \quad   \text{ for }  s \in S, \text{ and } 
\qquad
 \tilde{u}_t  := \frac{1-u_t}{J} \,\,\quad \text{ for }  t \in T.
\end{equation}
Using Eq.~ (\ref{seven}) in Eq.~(\ref{equationInRST}) we get for $i \notin S
\cup T$
\begin{subequations}\label{equationInRST_scaled}
\begin{equation} \label{eight}
\frac{du_i}{dt} =  \sum_{j  \notin S \cup T}w_{ij}u_j + \sum_{j \in T}w_{ij}
  + J\left(\sum_{s \in S}w_{is}\tilde{u}_s-\sum_{t \in
T}w_{it}\tilde{u}_t\right) +b_i,
\end{equation}
and for $s \in S$,
\begin{align} \label{nine}
J\frac{d\tilde{u}_s}{dt} = & \sum_{j \notin S \cup T}w_{sj}^+u_j + \sum_{t \in
T}w_{st}^+
  +J\left(\sum_{j \in S}w_{sj}^+\tilde{u}_j-\sum_{t \in
T}w_{st}^+\tilde{u}_t\right)+b_s^+ \nonumber\\
&-\left( \sum_{j \notin S \cup T }w_{sj}^-u_j + \sum_{t \in T}w_{st}^- +
b_s^-\right)\frac{\tilde{u}_s}{1+\tilde{u}_s}
  -J\left(\sum_{j \in S}w_{sj}^+\tilde{u}_j-\sum_{t \in
T}w_{st}^+\tilde{u}_t5\right)\frac{\tilde{u}_s}{1+\tilde{u}_s},
\end{align}
and similarly, for $t \in T$,
\begin{align} \label{ten}
J\frac{d\tilde{u}_t}{dt} = &-\left( \sum_{j \notin S \cup T }w_{tj}^+u_j +
\sum_{j \in T}w_{tj}^++b_t^+\right)\frac{\tilde{u}_t}{1+\tilde{u}_t}
 -J\left(\sum_{s \in S}w_{ts}^+\tilde{u}_s-\sum_{j \in
T}w_{tj}^+\tilde{u}_j\right)\frac{\tilde{u}_t}{1+\tilde{u}_t} \nonumber \\
&+\sum_{j \notin S \cup T}w_{tj}^-u_j + \sum_{j \in T}w_{tj}^- + b_t^-
+J\left(\sum_{s \in S}w_{ts}^+\tilde{u}_s-\sum_{j \in
T}w_{tj}^+\tilde{u}_j\right).
\end{align}
\end{subequations}

When $J$ is small, we can apply Geometric Singular Perturbation Theory (GSPT) to Eq. (\ref{equationInRST_scaled}) ~[\cite{Hek2010,Kaper1998}].  The GSPT posits that, under a normal hyperbolicity condition which we verify below, Eq.~(\ref{equationInRST_scaled}) can be further simplified by assuming that $J = 0$. 
This yields a 
differential-algebraic system
\begin{subequations}\label{mainReduced}
\begin{align}
\frac{du_i}{dt} &=  \sum_{j  \notin S \cup T}w_{ij}u_j + \sum_{j \in T}w_{ij} 
+b_i,
& i \notin S \cup T;  \label{eleven} \\
0 &=  \sum_{j \notin S \cup T}w_{sj}^+u_j + \sum_{t \in T}w_{st}^+
  +b_s^+ -\left( \sum_{j \notin S \cup T }w_{sj}^-u_j + \sum_{t \in T}w_{st}^- +
b_s^-\right)\frac{\tilde{u}_s}{1+\tilde{u}_s},
& s \in S; \label{twelve} \\
0 &= -\left( \sum_{j \notin S \cup T }w_{tj}^+u_j + \sum_{j \in
T}w_{tj}^++b_t^+\right)\frac{\tilde{u}_t}{1+\tilde{u}_t}
+\sum_{j \notin S \cup T}w_{tj}^-u_j + \sum_{j \in T}w_{tj}^- + b_t^-,
& t \in T. \label{thirteen}
\end{align}
\end{subequations}
which is equivalent to Eq.~\eqref{eqn:mainReduced_RST} after rescaling.  This
conclusion
will be justified if the manifold defined by Eqs.~(\ref{twelve}) and
(\ref{thirteen}) is normally hyperbolic and
stable~[\cite{Fenichel1979,Kaper1998,Hek2010}].  We verify this condition next.

Let $ \hat{u} =  \{ u_{i_1},...,u_{i_m} \}$ where $\{ i_1,...,i_m \} = \{
1,2,...,N \}\backslash (S \cup T) $, be the coordinates of $u$ which are away
from the boundary, and denote the right hand side of Eq.~(\ref{twelve}) by
$F_s(\hat{u}, \tilde{u}_{i_s})$, for all  $s \in S$, so that
\begin{align*}
 F_s(\hat{u}, \tilde{u}_{i_s}) :=  \sum_{j \notin S \cup T}w_{sj}^+u_j + \sum_{t
\in T}w_{st}^+
  +b_s^+ -\left( \sum_{j \notin S \cup T }w_{sj}^-u_j + \sum_{t \in T}w_{st}^- +
b_s^-\right)\frac{\tilde{u}_s}{1+\tilde{u}_s},
\end{align*}
and
\begin{align*}
 \frac{\partial F_s}{\partial \tilde{u}_{i_s}} =  -\left( \sum_{j \notin S \cup
T }w_{sj}^-u_j + \sum_{t \in T}w_{st}^- + b_s^-\right) \left(
\frac{1}{1+\tilde{u}_s}\right)^2 
< 0,
\end{align*}
for all $s \in S$.
Similarly, by denoting the right hand side of Eq.~(\ref{thirteen}) by
$G_t(\hat{u}, \tilde{u}_{i_t})$, for all  $t \in T$. \emph{i.e.}
\begin{align*}
 G_t(\hat{u}, \tilde{u}_{i_t}) := -\left( \sum_{j \notin S \cup T }w_{tj}^+u_j +
\sum_{j \in T}w_{tj}^++b_t^+\right)\frac{\tilde{u}_t}{1+\tilde{u}_t}
+\sum_{j \notin S \cup T}w_{tj}^-u_j + \sum_{j \in T}w_{tj}^- + b_t^-,
\end{align*}
we see that
\begin{align*}
 \frac{\partial G_t}{\partial \tilde{u}_{i_t}} = -\left( \sum_{j \notin S \cup T
}w_{tj}^+u_j + \sum_{j \in T}w_{tj}^++b_t^+\right) \left(
\frac{\tilde{u}_t}{1+\tilde{u}_t} \right)^2
< 0.
\end{align*}
Hence, the manifold defined by Eqs.~(\ref{twelve}) and (\ref{thirteen}) is
normally hyperbolic and stable.  This completes the proof that the reduction of
the non-linear system~(\ref{eqn:ProblemEquation}) to the solvable system given in Eq.~(\ref{eqn:mainReduced_RST}) is justified for small $J$.

\subsection{Boolean approximation}\label{sec:math_BN}

Here we formally show that the steady states of the BN given in Eq.~\eqref{eqn:mainReduced_BN} are in a one-to-one correspondence with the steady states of the system  given in Eq.~\eqref{eqn:ProblemEquation}. We also show that under some conditions trajectories in the BN correspond to trajectories of the  system  Eq.~\eqref{eqn:ProblemEquation}.

\subsubsection{Steady state equivalence at the corner subdomains}\label{sec:math_BN_ss_corner}

First we prove the one-to-one correspondence only at the corner subdomains using the PL approximation. We do this by showing that Eq.~\eqref{eqn:mainReduced_RST} has a steady state at a corner subdomain $\mathcal{R}_S^T$ if and only if the BN has a steady state at the corner $x\in\{0,1\}^N$ contained in $\mathcal{R}_S^T$.

We proceed from Eq.~(\ref{eqn:mainReduced_RST}) for a corner subdomain $\mathcal{R}_S^T$ so that $S\cup T=\{1,\ldots,N\}$. We obtain the equations

\[0=\sum_{t \in T}w_{st}^+  +b_s^+ -\left( \sum_{t \in T}w_{st}^- +b_s^-\right)\frac{u_s}{J+u_s},\qquad  s \in S \] and 
\[0= -\left( \sum_{j \in T}w_{tj}^++b_t^+\right)\frac{1-u_t}{J+1-u_t} +\sum_{j \in T}w_{tj}^- + b_t^-, \qquad  t \in T.\]

For the sets $S$ and $T$, consider $x\in\{0,1\}^N$ such that $x_s=0$ for all $s\in S$ and $x_t=1$ for all $t\in T$. Then, we can write the equations above in a more compact form

\[0=A_s(x) -I_s(x)\frac{u_s}{J+u_s}, \quad s \in S, \qquad \text{and,} \qquad
0=-A_t(x)\frac{1-u_t}{J+1-u_t} +I_t(x), \quad t \in T.\]

Solving these equations for $u_s$ and $u_t$, respectively, we obtain
\begin{equation} \label{E:fp_equation}
u_s=-\frac{A_s(x)J}{A_s(x)-I_s(x)}, \quad s \in S, \qquad  \text{and,}  
\qquad 
u_t=1-\frac{I_t(x)J}{A_t(x)-I_t(x)}, \quad t \in T.
\end{equation}

Now, let $\epsilon>0$ small such that $\left|\frac{I_t(x)J}{A_t(x)-I_t(x)}\right|, \left| \frac{I_t(x)J}{A_t(x)-I_t(x)} \right|\leq 1$. 
For all $J$ such that $0<J<\epsilon$, Eq.~\eqref{E:fp_equation} has a solution inside $[0,1]^N$ if and only if 
$A_s(x)-I_s(x)<0, s \in S,$ and $A_t(x)-I_t(x)>0, t \in T$,
or more compactly if and only if
$H(A_i(x)-I_i(x))=x_i \textrm{ for all } i=1,\ldots,N.$
Thus, a steady state appears in the corner subdomain corresponding to $x$ if and only if $x$ is a steady state of the BN $h=(h_1,\ldots,h_N):\{0,1\}^N\rightarrow\{0,1\}^N$ given by Eq.~\eqref{eqn:mainReduced_BN}. 
We have proved,

\begin{theorem}\label{thm:localss}
There is an $\epsilon>0$ such that for all $0<J<\epsilon$,  Eq.~\eqref{eqn:mainReduced_RST} has a steady state at a corner subdomain containing $x\in\{0,1\}^N$ if and only if the BN in Eq.~\eqref{eqn:mainReduced_BN} has a steady state at $x$.
\end{theorem}

Also, since we showed that the PL system given by  Eq.~\eqref{eqn:mainReduced_RST} is a valid approximation of the full system given in Eq.~\eqref{eqn:ProblemEquation}, we have the following corollary.

\begin{corollary}\label{cor:localss}
There is an $\epsilon>0$ such that for all $0<J<\epsilon$, the system in Eq.~\eqref{eqn:ProblemEquation} has a steady state at a corner subdomain containing $x\in\{0,1\}^N$ if and only if the BN in Eq.~\eqref{eqn:mainReduced_BN} has a steady state at $x$.
\end{corollary}

\subsubsection{Global equivalence of steady states}
\label{sec:math_BN_ss_global}

We proved that, for $J$ small, the steady states at the corner subdomains of the system in Eq.~(\ref{eqn:ProblemEquation}) are in a one-to-one correspondence with the steady states of the BN. However, the corner subdomains only cover a small portion of $[0,1]^N$.  We next show that the steady state correspondence is global.

Recall that $\Omega=\{\mathcal{C}: \mathcal{C} \textrm{ is a chamber}\}$ and that  each chamber is  a connected component of $[0,1]^N\setminus \cup_{i=1}^N\{u:\sum_{j=1}^N w_{ij} u_j +b_i=0\}$. The chambers are bounded convex open subsets of $[0,1]^N$ (with the topology inherited from $[0,1]^N$). Also note that $[0,1]^N\setminus \cup_{\mathcal{C} \in\Omega}\mathcal{C}$ is a finite union of hyperplanes and hence has measure zero. We denote by $\mathcal{C}(x)$ the chamber that contains the point $x \in [0,1]^N$.

\begin{theorem}\label{thm:one2one}
Let $K$ be a compact subset of $\cup_{\mathcal{C}\in\Omega}\mathcal{C}$ such that $K\cap \mathcal{C}$ is convex for any $\mathcal{C}\in\Omega$, and such that $K\cap \mathcal{C}(x)$ contains a neighborhood of $x$ for each $x\in \{0,1\}^N$. 

Then, there is an $\epsilon_K>0$ such that for all $0<J<\epsilon_K$,  
if $x^*\in\{0,1\}^N$ is a steady state of the BN in Eq.~\eqref{eqn:mainReduced_BN}, then the ODE in Eq.~\eqref{eqn:ProblemEquation} has a steady state $u^*\in \mathcal{C}(x^*)$; and if $u^*\in K$ is a steady state of the ODE, then $u^*\in \mathcal{C}(x^*)$ for some steady state $x^*$ of the BN. 

Furthermore, if $x^*$ is a steady state of the BN in Eq.~\eqref{eqn:mainReduced_BN}, then the steady state of the ODE in Eq.~\eqref{eqn:ProblemEquation} is unique in $K\cap \mathcal{C}(x^*)$, converges to $x^*$ as $J\rightarrow 0$, and is asymptotically stable.
\end{theorem}
\begin{proof}
See Appendix.
\end{proof} 

We can make the set $K$ in Theorem \ref{thm:one2one} as close to $[0,1]^N$ as desired. For example, for each chamber $\mathcal{C}$ and for $r>0$, denote $K_\mathcal{C}:=\{u\in\mathcal{C}:|\sum_{j=1}^N w_{ij}u_j+b_i|\geq r, \forall i\}$. By using $r$ small, and denoting Lebesgue measure by $\mu$, we can make $\mu([0,1]^N\setminus \cup_{\mathcal{C}\in\Omega}K_{\mathcal{C}})=\mu( \cup_{\mathcal{C}\in\Omega} (\mathcal{C}\setminus K_{\mathcal{C}}))= \sum_{\mathcal{C}\in\Omega} \mu(\mathcal{C}\setminus K_{\mathcal{C}})$  as small as desired. Hence, we have the following corollary.

\begin{corollary}\label{cor:one2one}
For any $\epsilon>0$, there is a set $K\subseteq [0,1]^N$ satisfying $\mu([0,1]^N\setminus K)<\epsilon$ and a number $\epsilon_K$ such that for all $0<J<\epsilon_K$, there is a one-to-one correspondence between the steady states of the BN in Eq.~\eqref{eqn:mainReduced_BN} and the steady states in $K$ of the ODE in Eq.~\eqref{eqn:ProblemEquation}. Furthermore, the steady states of the ODE in $K$ are asymptotically stable.
\end{corollary}

Note that the set $K$ does not include the nullclines.  Hence,  steady states  outside $K$ are possible, and they could be stable. Such steady states can be studied using the PL approximation in Eq.~\eqref{eqn:mainReduced_RSTa}.

\subsubsection{Equivalence of trajectories}
\label{sec:math_BN_trajectories}

We next examine under which conditions  the trajectories of the BN in Eq.~(\ref{eqn:mainReduced_BN}) correspond trajectories of
the ODE given in Eq.~(\ref{eqn:ProblemEquation}). The main assumption in the rest of this section is that the hyperplanes divide $[0,1]^N$ into $2^N$ chambers, and each chamber contains a corner. 

We say that the solutions of the system in Eq.~(\ref{eqn:ProblemEquation})  \emph{transition from $K\subseteq[0,1]^N$ to $K'\subseteq[0,1]^N$} if for any solution of the system, $u(t)$, with initial condition $u(0)\in K$, there exists $\hat{t}$ such that $u(\hat{t})\in K'$. The \emph{Hamming distance} between $x,y\in\{0,1\}^N$ is defined as $d(x,y):=|\{i\in\{1,\ldots,N\}:x_i\neq y_i\}|$. We denote a transition $h(x)=y$ (i.e. $H(Wx+b)=y$) with $x\rightarrow y$. A trajectory is a sequence of transitions and is denoted similarly.
We call a transition $x\rightarrow y$ \emph{regular} if either (1) $x=y$ or (2) $d(x,y)=1$ and $h_j(y)=y_j$ for the index $j$ for which $x_i\neq y_i$. 
In other words, a transition  $x\rightarrow y$ in the BN is regular if $x$ is a steady state or if $x$ transitions to $y$ by changing only one coordinate and this coordinate does not change back when transitioning from $y$ to $h(y)$. For example, if $000\rightarrow 100\rightarrow 111$, then $000\rightarrow 100$ is a regular transition ($j=1$); on the other hand, if $000\rightarrow 100\rightarrow 010$, then $000\rightarrow 100$ is not a regular transition. A trajectory is regular if each component transition is regular.

\begin{theorem}\label{thm:transition}
Suppose that the hyperplanes divide $[0,1]^N$ into $2^N$ chambers and consider a regular transition of the BN in Eq.~\eqref{eqn:mainReduced_BN}, $x\rightarrow h(x)$. Then, there is a neighborhood $K$ of $x$, and an $\epsilon_K>0$ such that for all $0<J<\epsilon_K$ the solutions of the ODE in Eq.~\eqref{eqn:ProblemEquation} transition from $K$ to $\mathcal{C}(h(x))$. Also, if $x$ is a steady state, the neighborhood $K$ can be chosen to be invariant.
\end{theorem}

\begin{proof}
See Appendix.
\end{proof}

Next for a steady state $x$ define $B_1(x)=\{y\in\{0,1\}^N:h(y)=x\textrm{ and } d(x,y)\leq 1\}$; that is, $B_1(x)$ is the set of states in the basin of attraction of $x$ with Hamming distance at most 1 from $x$. The following corollary of Theorem \ref{thm:transition} states that this part of the basin of attraction of $x$ corresponds to part of the basin of attraction of the steady state in the ODE in Eq.~(\ref{eqn:ProblemEquation})  corresponding to $x$.
\begin{corollary}
Suppose that $x$ is a steady state of the BN in Eq.~\eqref{eqn:mainReduced_BN}. Then, for every $y\in B_1(x)$, there is a neighborhood $K$ of $y$  and $\epsilon_K>0$, such that for all $0<J<\epsilon_K$  the solutions of the ODE in Eq.~\eqref{eqn:ProblemEquation} transition from $K$ to  $\mathcal{C}(x)$.
\end{corollary}

Note that Theorem \ref{thm:transition} implies that for a regular trajectory of the BN in Eq.~(\ref{eqn:mainReduced_BN}),  $x\rightarrow h(x)\rightarrow h^2(x)\rightarrow \ldots \rightarrow h^m(x)$, the solutions of the ODE in Eq.~(\ref{eqn:ProblemEquation}) will transition from a neighborhood of $x$ to $\mathcal{C}(h(x))$, from a neighborhood of $h(x)$ to $\mathcal{C}(h^2(x))$ and so on. 
To guarantee that a neighborhood of $x$ will reach a neighborhood of $h^m(x)$ (that is, to guarantee that the result is ``transitive''), we need the additional assumption that each hyperplane is orthogonal to some coordinate axis. Note that the example given in Section \ref{sec:3nodes} satisfies this condition.

\begin{theorem}\label{thm:trajectory}
Suppose that  each hyperplane is orthogonal to some coordinate axis and let  $x\rightarrow h(x)\rightarrow \ldots\rightarrow h^m(x)$ be a regular trajectory of the BN in Eq.~(\ref{eqn:mainReduced_BN}). Then, for any compact set $K\subset \mathcal{C}(x)$ there is $\epsilon_{K}>0$ such that for all $0<J<\epsilon_{K}$, the solutions of the ODE in Eq.~\eqref{eqn:ProblemEquation} transition from $K$ to $\mathcal{C}(h^m(x))$ following the order of the regular trajectory.
\end{theorem}

\begin{proof}
See Appendix.
\end{proof}

For a steady state of the BN define $B(x)=\{y:\textrm{ there is a regular trajectory from $y$ to $x$}\}$. The following corollary of Theorem~\ref{thm:trajectory} implies that some states in the basin of attraction of a steady state of the BN in Eq.~(\ref{eqn:mainReduced_BN}) correspond to chambers in the basin of attraction of a steady state of the ODE in Eq.~(\ref{eqn:ProblemEquation}).

\begin{corollary} 
Suppose that  each hyperplane is orthogonal to some coordinate axis and let $x$ be a steady state of the BN in Eq.~\eqref{eqn:mainReduced_BN}. Consider $y\in B(x)$. Then, for any compact set $K\subseteq\mathcal{C}(y)$, there exists $\epsilon_{K}>0$ such that for all $0<J<\epsilon_{K}$, the solutions of the ODE in Eq.~\eqref{eqn:ProblemEquation} transition from $K$ to $\mathcal{C}(x)$.
\end{corollary}

Similarly, we obtain the following corollary for oscillatory behavior. 
\begin{corollary}\label{cor:oscilation}
Suppose that  each hyperplane is orthogonal to some coordinate axis and let  $x^1\rightarrow x^2\rightarrow \ldots\rightarrow x^p\rightarrow x^1$ be a regular periodic orbit of the BN in Eq.~(\ref{eqn:mainReduced_BN}). Then, for any compact set $K\subseteq\mathcal{C}(x^1)$ and any positive integer $m$, there exists $\epsilon_{K,m}>0$ such that for all $0<J<\epsilon_{K,m}$, the solutions of the ODE in Eq.~\eqref{eqn:ProblemEquation} transition between the chambers (starting at $K$) in the order $\mathcal{C}(x^1), \mathcal{C}(x^2),\ldots,\mathcal{C}(x^p),\mathcal{C}(x^1)$, $m$ times.
\end{corollary}

Note that the example in Section \ref{sec:3nodes} satisfies the hypothesis of this last corollary. In general, Corollary \ref{cor:oscilation} does not guarantee that the solution is periodic.

Finally, we note that the requirement that there are $2^N$ chambers, each containing a corner is necessary.  Even if we have a transition where only one variable changes (e.g. $h(1,0)=(0,0)$), having an intermediate chamber can change the behavior of the solutions before they reach the chamber predicted by the BN. In the example shown in Fig.~\ref{fig:twoDimension_counterexample_regular} the signs of the vector field of the approximating linear system imply that the BN transitions from $(1,0)$ to $(0,0)$.  However, solutions can transition from the chamber that contains $(1,0)$ to the bottom middle chamber and never reach the chamber that contains $(0,0)$. In summary, even having a transition where only one variable changes may not be sufficient to guarantee that the Boolean transition corresponds to a similar transition in the original system.

\begin{figure}[h]
\begin{center}
\includegraphics[scale = .35]{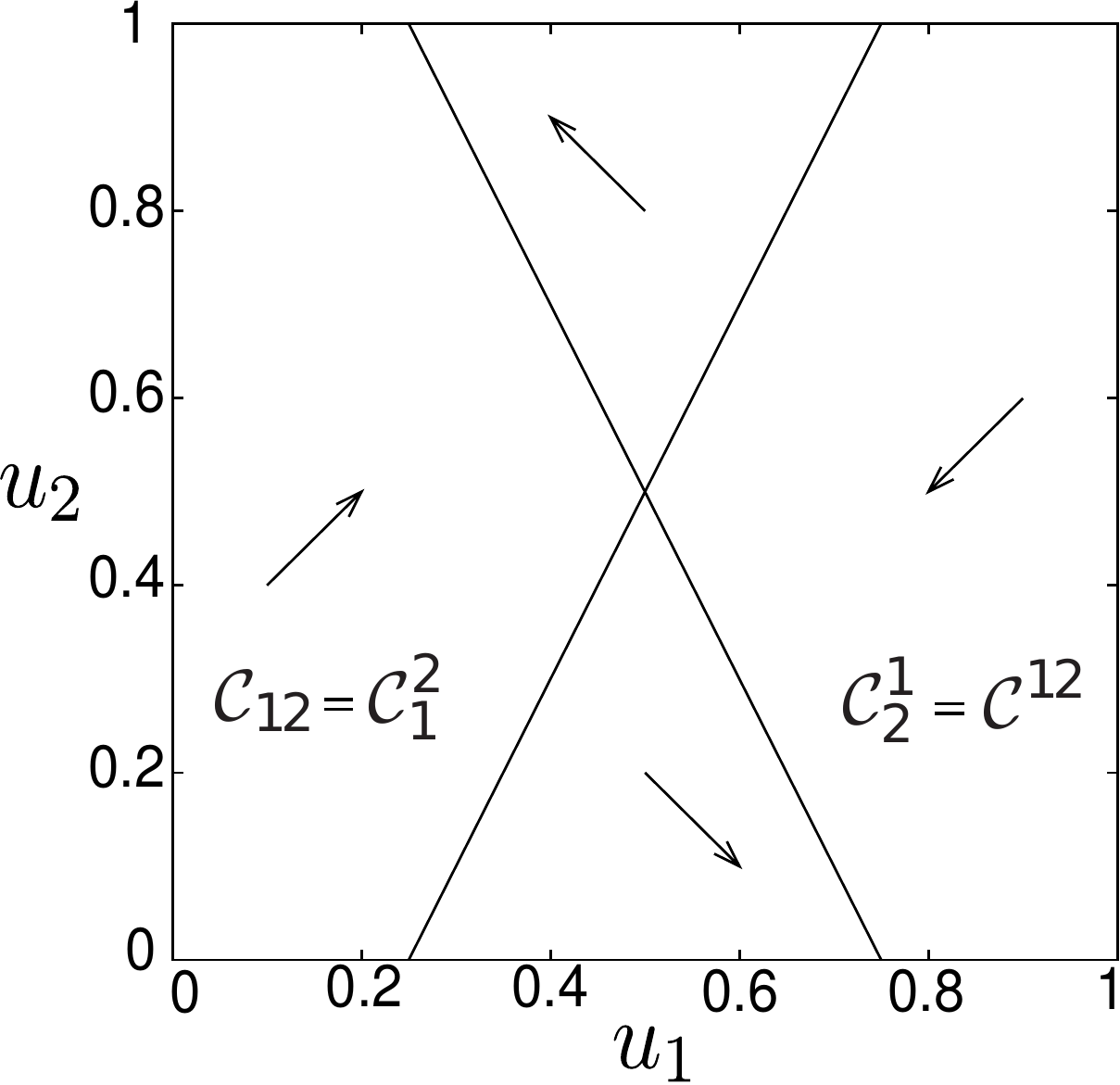}
\parbox{.9\textwidth}{
\caption[]{\footnotesize
Chambers and signs of vector field for the linear system given by $W=\left[\begin{array}{rr} -2 & -1 \\ -2 & 1 \end{array}\right]$ and $b=\left[\begin{array}{r} 1.5 \\ 0.5 \end{array}\right]$.
\label{fig:twoDimension_counterexample_regular}}
}
\end{center}
\end{figure}

\section{Discussion}\label{Discussion}

Models of biological systems are frequently nonlinear and difficult to analyze mathematically. In addition, accurate models 
frequently contain numerous parameters whose exact values are not known. Thus, studying which parameters have a large impact on dynamics, and how a  model can be simplified, is crucial in finding the features of biological systems that determine their behavior and responses. Reduction techniques that preserve key dynamical properties are essential in this endeavor.

We studied a special class of non-linear differential equation  models of biological networks where interactions between nodes are described using Hill functions. When the Michaelis-Menten constants are sufficiently small, the behavior of the  system is captured by an approximate piecewise linear system and a Boolean Network. 
In this case the domain of the full system naturally decomposes into nested hypercubes.  These hypercubes define subdomains within which a solvable linear--algebraic system approximates the original system. The Boolean Network is obtained from a decomposition of the domain into chambers and describes how solutions evolve between them.

The proposed reductions have a  number of advantages: The piecewise linear approximation is not only  easier to solve than the original system analytically, but also numerically (the original system becomes stiff for small $J$). When  one is interested in qualitative behavior such as steady state analysis, the Boolean approximation can be very useful, especially when the dimension of the system is large. Also, the Boolean framework has been used to model many biological systems where it is assumed that interactions are switch-like and variables can be discretized. It is therefore important to know when, and in what sense such reduced systems can be justified, and in particular, what dynamical properties of the full system are capture by a reduction. 

Although the case of large exponent $n$ in the Hill function  has been studied in the past [\cite{GlassKauffman1973,Snoussi1989, thomasbook,mendoza2006, DavitichBornholdt2008, Wittmann2009, Franke2010, Veliz-CubaArthurHochstetlerKlompsKorpi2012}], the case of small $J$ has been studied only recently and heuristically [\cite{DavitichBornholdt2008}]. In this manuscript we have shown that the PL and the BN approximations are also valid for the case of small $J$, and have given explicit formulas for their computation. The BN approximation preserves steady state behavior and under further restrictions, it can also be used to  infer the basins of attractions and oscillations in the original system.  Note that the Boolean functions in the Boolean approximation are threshold functions, as used in earlier models [\cite{Li_cc_2004,DavitichBornholdt2008b}].  Our results show that such BN scan indeed appear when  approximating more detailed models, such as those described by Eq.~\eqref{eqn:ProblemEquation}. 
In summary, our results for the limit $J\rightarrow 0$ complement previous results for the limit $n\rightarrow \infty$,  providing a useful framework for reducing nonlinear systems to PL systems and BNs.

A potential limitation in our arguments is that we have an approximation 
valid only in an asymptotic limit.  It is unknown when and how the approximation breaks down.
However, the approximation is still valid as $J$ increases until we reach a bifurcation point, which can happen for relative large $J$ or for values of $J$ that are biologically relevant. Also, we have not provided a systematic relationship between the thickness of the boundary, $\d$, and the Michaelis-Menten constant, $J$. Numerical tests suggest that the relationship is between $J = \mathcal{O}(\d)$ and $J = \mathcal{O}(\d^2)$. Another limitation of our analysis is that, in the general case, it is not known when the transitions or cycles in a BN will correspond to similar transitions or cycles in the original ODE.

\section*{Acknowledgment}

We thank Matthew Bennett for help in preparing the manuscript. This work was funded by the NIH, through the joint NSF/NIGMS Mathematical Biology Program Grant No. R01GM104974.


\bibliographystyle{plainnat}

\section{Appendix}

\subsection{Motivation of Eq.~\eqref{eqn:ProblemEquation}}
Here we present a heuristic justification of the use of Eq.~\eqref{eqn:ProblemEquation}. The ideas follow those presented in [\cite{GoldbeterKoshland1981, Goldbeter1991, Novak1998, NovakPatakiCilibertoTyson2001, NovakPatakiCilibertoTyson2003,  Aguda2006}]. As mentioned in the Introduction, this is only heuristic in general.

Consider a protein that can exist in an unmodified form, $W,$ and a modified form, $W^*,$ where the conversion between the two forms is catalyzed by two enzymes, $E_1$ and $E_2$. That is, consider the reactions
\[
W+E_1 \mathop{\rightleftarrows}_{k_{-1}}^{k_1} WE_1 \mathop{\rightarrow}^{p_1} W^*+E_1,
\]
\[
W^*+E_2 \mathop{\rightleftarrows}_{k_{-2}}^{k_2} W^*E_2 \mathop{\rightarrow}^{p_2} W+E_2.
\]
Then, using quasi-steady-state assumptions  one can obtain the equation

\[
\frac{dW^*}{dt}=A\frac{L -W^*}{K_1+L-W^*}-I\frac{W^*}{K_2+W*},
\]
where $A$, $I$, $L$, $K_1$, $K_2$ depend on $k_1$, $k_{-1}$, $p_1$, $k_2$, $k_{-2}$, $p_2$, $[E_1]$, $[WE_1]$, $[E_2]$, and  $[W^*E_2]$ [\cite{GoldbeterKoshland1981}]. After rescaling by $L$ we obtain Eq.~\eqref{eqn:ProblemEquation}.

Now, consider a system with $N$ species (e.g. proteins) and assume that $u_i(t)$ and $v_i(t)$ represent the  concentration of species $i$ at time $t$ in its active and inactive form, respectively. Furthermore, suppose that the total concentration of each species is constant and that the difference between decay and production is negligible (so that $u_i(t)+v_i(t)$ is constant). That is, \[u_i(t)+v_i(t)=L_i,\]
 where $L_i$ does not depend on time, and 
 \[\frac{du_i}{dt}=\textrm{rate of activation} - \textrm{rate of inhibition}  .\]
  Then, using Michaelis-Menten kinetics, the rate of activation of this species can be modeled by 
  \[\textrm{rate of activation}= A_i\frac{v_i}{K_{i}^A+v_i}=A_i\frac{L_i -u_i}{K_{i}^A+L_i-u_i},\]
 where the maximal rate, $A_i=A_i(u)$, is a function of the different  species in the network. Similarly, modeling the inhibition of the species using Michaelis-Menten kinetics, we obtain
  \[\textrm{rate of inhibition}= I_i\frac{u_i}{K_{i}^A+u_i}.\]
  
  Thus, we obtain 
  \[\frac{du_i}{dt}=A_i\frac{L_i -u_i}{K_{i}^A+L_i-u_i}-I_i\frac{u_i}{K_{i}^A+u_i}.\]

Now, we rescale $u_i\rightarrow L_iu_i$, $A_i\rightarrow L_iA_i$, $I_i\rightarrow L_iI_i$ and obtain 
  \[\frac{du_i}{dt}=A_i\frac{1 -u_i}{\frac{K_{i}^A}{L_i}+1-u_i}-I_i\frac{u_i}{\frac{K_{i}^A}{L_i}+u_i}.\]
  
  Hence, by denoting $J_i^A:=\frac{K_{i}^A}{L_i}$ and $J_i^A:=\frac{K_{i}^A}{L_i}$, we obtain the system given in Eq.~\eqref{eqn:ProblemEquation}. Also, $J_i^A$ and $J_i^I$ small means that the  dissociation constants ($K_i^A, K_i^I$) are much smaller than the total concentration of species $i$; that is, $J_i^A, J_i^I\ll 1$ if and only if $K_i^A, K_i^I\ll L_i$. Note that the initial conditions now satisfy $u_i(0)
\in [0,1]$ for all $i$.

\subsection{Behavior of $A\frac{1-x}{J+1-x}-I\frac{x}{J+x}$ as $J\rightarrow 0$}

Consider the one-dimensional system 
\begin{equation}\label{MMintro}
 \frac{dx}{dt} = A \frac{1-x}{J+1-x} - I \frac{x}{J+x}.
\end{equation}

Fig.~\ref{fig:MMintro} shows the graph of the right hand side of Eq.~(\ref{MMintro}) for the fixed values  $A = 1, I = 0.5$ and three different values of $J$. Note that as $J$ becomes smaller, the graph gets flatter in (0,1). Then, for $J$ small, we can approximate Eq.~(\ref{MMintro}) in the interior of the region $[0,1]$ by the linear ODE 
$$\frac{dx}{dt} = A-I.$$
For $x\sim J$, we can approximate Eq.~(\ref{MMintro})  by the ODE 
$$\frac{dx}{dt} = A-I\frac{x}{J+x}.$$
And for $x\sim 1-J$, we can approximate Eq.~(\ref{MMintro})  by the ODE 
$$\frac{dx}{dt} = A\frac{1-x}{J+1-x}-I.$$
\begin{figure}[h]
\begin{center}
\includegraphics[scale = .5]{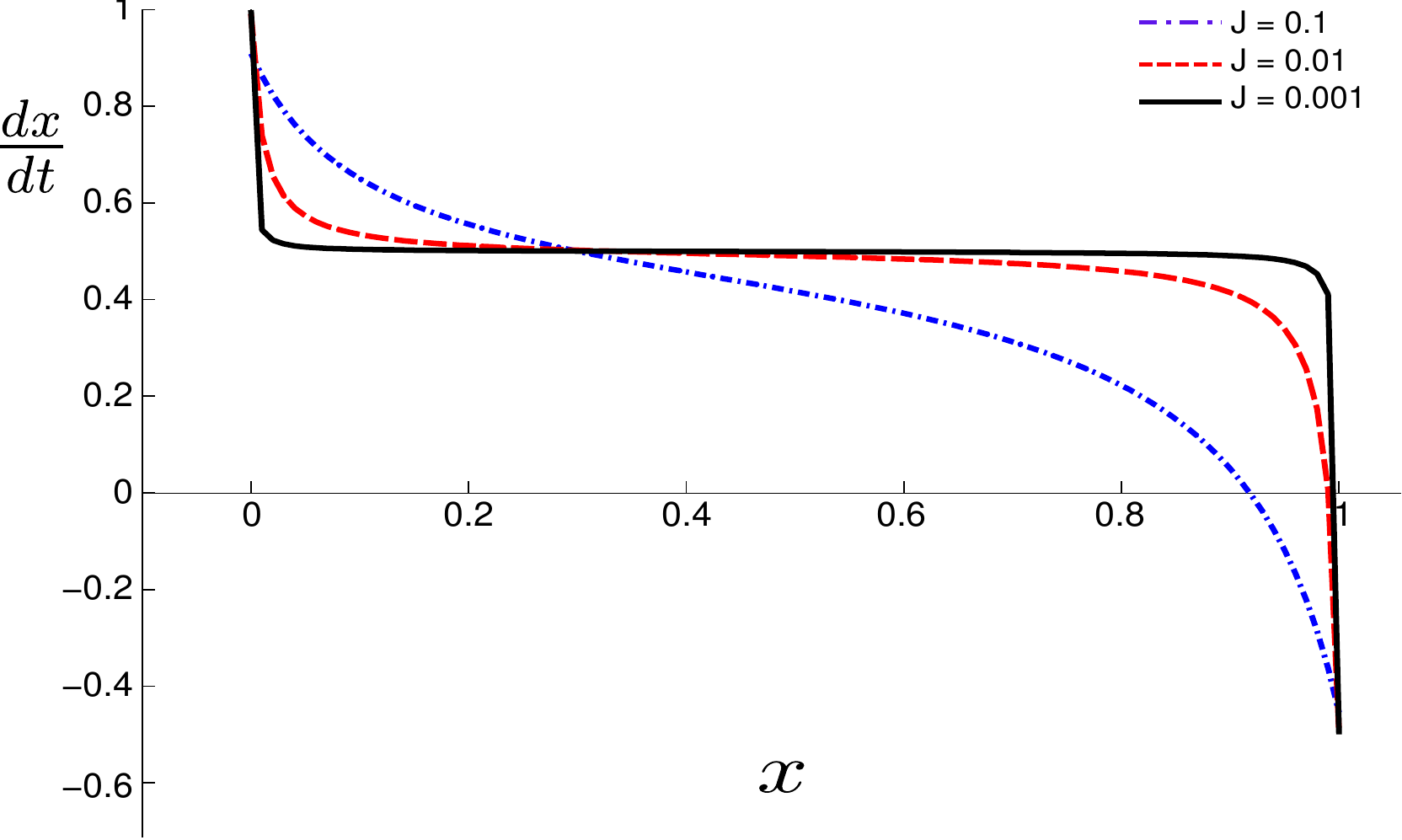}
\parbox{.9\textwidth}{
\caption{\footnotesize {Plots of right hand side of Eq.~(\ref{MMintro})} for three different values of $J$, as functions of $x$. Other parameters: $A = 1, I = .5$. This figure suggests that differential equations of the form Eq.~(\ref{MMintro}) can be approximated by linear ODEs in the interior of the domain.
}\label{fig:MMintro}}
\end{center}
\end{figure}

For the values $A=1$, $I=0.5$ we obtain the following approximations.\\
$\frac{dx}{dt} = 0.5$, if $J\ll x \ll 1-J$\\
$\frac{dx}{dt} = 1-0.5\frac{x}{J+x}$, if $x\sim J$\\
$\frac{dx}{dt} = \frac{1-x}{J+1-x}-0.5$, if $x\sim 1-J$

Note that there is an asymptotically stable steady state close to 1. Intuitively, for $J$ small, solutions that start in the region $x\sim J$ quickly reach the region $J\ll x \ll 1$, which behaves like a linear system. Then, solutions increase almost linearly (with slope 0.5) until they enter the region $x\sim 1-J$ where they will approach the steady state (see Fig.~\ref{fig:1D}).

\begin{figure}[h]
\begin{center}
\includegraphics[scale = .8]{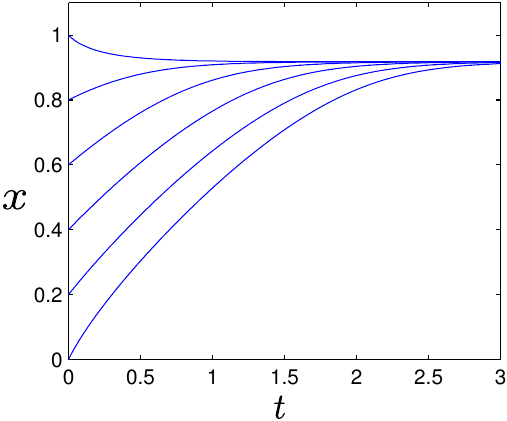}\quad
\includegraphics[scale = .8]{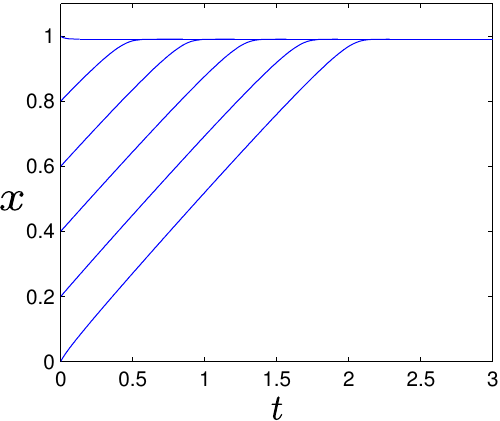}\quad
\includegraphics[scale = .8]{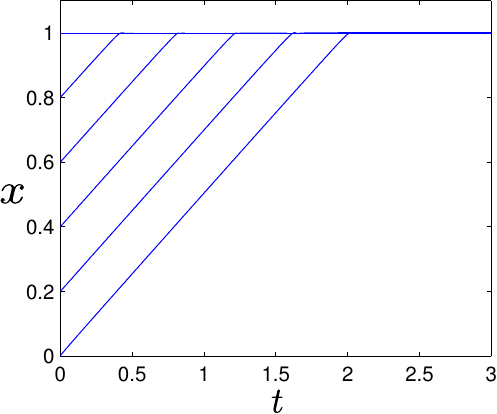}
\parbox{.9\textwidth}{
\caption{\footnotesize Solutions of  Eq.~(\ref{MMintro}) for three different values of $J$ (left: $J=0.1$, center: $J=0.01$, right: $J=0.001$). Other parameters: $A = 1, I = .5$. 
}\label{fig:1D}}
\end{center}
\end{figure}

We see that in the limit $J\rightarrow 0$ we obtain the solutions\\
$$x(t)=x(0)+0.5t \textrm{ for } t\in\left[0,\frac{1-x(0)}{0.5}\right]  \textrm{ and } x(t)=1 \textrm{ for } t\in\left[\frac{1-x(0)}{0.5},\infty\right),$$
where $x(0)\in[0,1]$.

\subsection{Proof of Theorem \ref{thm:one2one}}
The main idea in the proof is to use the fact that for $u_i\neq 0, 1$, the right-hand side of Eq.~\ref{eqn:ProblemEquation} converges to $Wu+b$. More precisely, the convergence is  uniform on compact subsets of $(0,1)^N$; so that  we have compact convergence.

Also, given a steady state of  Eq.~\eqref{eqn:ProblemEquation}, we can solve for $u_i$ and obtain $u_i=\Gamma^J_i(u)$ ($\Gamma^J_i$ will be defined later). The proofs also use the fact that as $J\rightarrow 0$, $\Gamma^J=(\Gamma^J_1,\ldots,\Gamma^J_N)$ converges uniformly to the function $u\mapsto H(Wu+b)$ on compact subsets of each chamber. That is, we also have compact convergence of $\Gamma^J$.

To prove Theorem \ref{thm:one2one} we need the following definitions and lemmas.

A point $u^*\in[0,1]^N$ will be a steady state of the ODE in Eq.~(\ref{eqn:ProblemEquation}) if and only if 
\begin{equation}\label{eqn:ss_mainODE}
A_i(u^*) \frac{1-u^*_i}{J+1-u^*_i}-I_i(u^*) \frac{u^*_i}{J+u^*_i}=0
\end{equation}
 for all $i$. Solving the corresponding quadratic equation for $u^*_i$ we obtain the solutions $u^*_i=1/2$ if $A_i(u^*)= I_i(u^*)$; and 
 
\begin{equation}\label{eqn:ss_quadformula}
u^*_i=\frac{(A_i(u^*)-I_i(u^*)-A_i(u^*)J-I_i(u^*)J)\pm\sqrt{\Delta_i(u^*)}}{2(A_i(u^*)-I_i(u^*))}
\end{equation}
if $A_i(u^*)\neq I_i(u^*)$, where $\Delta_i(u^*)$ is the discriminant of the quadratic equation, given by
\[\Delta_i(u):=(A_i(u)-I_i(u)-A_i(u)J-I_i(u)J)^2+4A_i(u)J(A_i(u)-I_i(u)).\]

The following lemma states that up to a set of small measure and for $J$ small all steady states of the ODE in Eq.~(\ref{eqn:ProblemEquation}) are given by the fixed points of the function $\Gamma^J=(\Gamma^J_1,\ldots,\Gamma^J_N)$ defined by

\begin{equation}\label{eqn:gamma}
\Gamma_i^J(u):=
\frac{(A_i(u)-I_i(u)-A_i(u)J-I_i(u)J)+\sqrt{\Delta_i(u)}}{2(A_i(u)-I_i(u))}.
\end{equation}

\begin{lemma}\label{lemma:well_defined}
For any compact subset $K$ of $\cup_{\mathcal{C}\in\Omega}\mathcal{C}$, there is an $\epsilon_K>0$ such that for all  $0<J<\epsilon_K$ the function $\Gamma^J$ is well-defined (as a real-valued function) on $K$, and
$u^*$ is a steady state in $K$ of the ODE in Eq.~(\ref{eqn:ProblemEquation}) if and only if $\Gamma^J(u^*)=u^*$.
\end{lemma}
\begin{proof}
Since the denominator is not 0 on $K$, we need to show that $\Delta_i(u)$ is non negative on $K$ for $J$ small.

Since $K$ is compact and $A_i(u)-I_i(u)=\sum_{j=1}^N w_{ij} u_j +b_i\neq 0$ for all $i=1,\ldots,N$ and for all $u\in K$, there is $r>0$ such that $|A_i(u)-I_i(u)|=|\sum_{j=1}^N w_{ij} u_j +b_i|\geq r$ on $K$. Then, since $(A_i(u)-I_i(u)-A_i(u)J-I_i(u)J)^2+4A_i(u)J(A_i(u)-I_i(u))$ converges uniformly as $J \rightarrow 0$ to $(A_i(u)-I_i(u))^2\geq r^2$ on $K$, there is $\epsilon_K>0$ such that for all $0<J<\epsilon_K$ the function $(A_i(u)-I_i(u)-A_i(u)J-I_i(u)J)^2+4A_i(u)J(A_i(u)-I_i(u))$ is positive on $K$ for all $i$. Thus, $\Gamma^J$ is well-defined for all $0<J<\epsilon_K$. 

If $u^*\in K$ and $\Gamma^J(u^*)=u^*$, then $u^*$ satisfies Eq.~(\ref{eqn:ss_quadformula}), and hence it is a steady state of the ODE in Eq.~(\ref{eqn:ProblemEquation}). Also, if $u^*$ is a steady state of the ODE in Eq.~(\ref{eqn:ProblemEquation}), then $u^*$ satisfies Eq.~(\ref{eqn:ss_quadformula}). However, only $\Gamma_i^J(u^*)$ is in $[0,1]$ and hence $u^*=\Gamma^J(u^*)$.
\end{proof}

It is important to notice that if $A_i(u)-I_i(u)>0$ then $\Gamma_i^J(u)$ (which is well-defined for $J$ small) converges to 1 and if $A_i(u)-I_i(u)<0$ then $\Gamma_i^J(u)$ converges to 0  as $J\rightarrow 0$. Hence, $\Gamma^J(u)$ converges pointwise to $H(Wu+b)$ on $\cup_{\mathcal{C}\in\Omega}\mathcal{C}$ as $J\rightarrow 0$. The next lemma states that for any compact subset of $\cup_{\mathcal{C}\in\Omega}\mathcal{C}$ we have uniform convergence and that the derivative of this function converges uniformly to zero.

\begin{lemma}\label{lemma:conv2H}
If $K$ is a compact subset of $\cup_{\mathcal{C}\in\Omega}\mathcal{C}$, then
\begin{itemize}
\item The function $\Gamma^J$ converges uniformly to the function $u\mapsto H(Wu+b)$ on $K$ as $J\rightarrow 0$. In particular, since $H(Wu+b)$ is constant in each chamber, we get that for any chamber $\mathcal{C}\in\cup_{\mathcal{C}\in\Omega}\mathcal{C}$, $\Gamma^J$ converges uniformly to the constant function $H(Wv+b)$ on $K\cap \mathcal{C}$ for any fixed $v\in C$.  Also, $\Gamma^J$ converges uniformly to the constant function $H(Wx+b)$ on $K\cap \mathcal{C}(x)$ for any $x\in \{0,1\}^N$ 
\item The Jacobian matrix $D \Gamma^J$ converges uniformly to zero on $K$ as $J\rightarrow 0$.  
\end{itemize} 
\end{lemma}
\begin{proof}
Similar to the proof of Lemma \ref{lemma:well_defined}, there is a number $r>0$ such that $|A_i(u)-I_i(u)|\geq r$  for all $i$ and for all $u\in K$, which is enough to guarantee uniform convergence on $K$.
\end{proof}

We now prove Theorem \ref{thm:one2one}
\begin{proof}
In this proof, ``ODE'' will refer to the ODE in Eq.~(\ref{eqn:ProblemEquation}) and ``BN'' will refer to the BN in Eq.~(\ref{eqn:mainReduced_BN}). Even though the steady states of this ODE depend on $J$, for simplicity we will denote them by $u^*$ instead of $u^{*J}$. 

First, from Lemma \ref{lemma:well_defined} we consider $\epsilon_K>0$ such that for all $0<J<\epsilon_K$ the function $\Gamma^J$ is well-defined on $K$ and such that $u^*\in K$ is a steady state of the ODE if and only if $u^*$ is a fixed point of $\Gamma^J$. Second, from Lemma \ref{lemma:conv2H} we have that $\Gamma^J$ converges uniformly to the constant vector $h(x)=H(Wx+b)$ on $K\cap \mathcal{C}(x)$. Since $h(x)$ is in $\{0,1\}^N$, we have that $K\cap \mathcal{C}(h(x))$ contains a neighborhood of $h(x)$. It follows from uniform convergence  that for all $0<J<\epsilon_K$ (taking a smaller $\epsilon_K$ if necessary) , $\Gamma^J(K\cap \mathcal{C}(x))\subseteq K\cap \mathcal{C}(h(x))$. Also, on any chamber $\mathcal{C}$ that does not contain an element of $\{0,1\}^N$, $\Gamma^J$ converges uniformly to the constant vector $x:=H(Wv+b)$ for any fixed $v\in \mathcal{C}$ ($H(Wu+b)$ is constant on $C$); then, for all $0<J<\epsilon_K$ (taking a smaller $\epsilon_K$ if necessary)  $\Gamma^J(K\cap \mathcal{C})\subseteq K\cap \mathcal{C}(x)$. 
Also note that $K\cap \mathcal{C}$ is compact for any chamber $\mathcal{C}$.

Now, suppose $x^*$ is a steady state of the BN, that is, $h(x^*)=x^*$. Then, for all $0<J<\epsilon_K$ we obtain that $\Gamma^J(K\cap \mathcal{C}(x^*))\subseteq K\cap \mathcal{C}(h(x^*))=K\cap \mathcal{C}(x^*)$. Since we have a continuous function from a convex compact set to itself, $\Gamma^J$ has a fixed point  $\Gamma^J(u^*)=u^*\in K\cap \mathcal{C}(x^*)$. Then, $u^*\in K\cap \mathcal{C}(x^*)$ is a steady state of the ODE. Now suppose that the ODE has a steady state $u^*\in K$, and let $\mathcal{C}$ be the chamber that contains $u^*$ and  $x^*:=H(Wu^*+b)$. Since $u^*=\Gamma^J(u^*)\in\Gamma^J(K\cap \mathcal{C})\subseteq K\cap \mathcal{C}(x^*)$, we have that $u^*\in \mathcal{C}(x^*)$. Since $u^*$ and $x^*$ belong to the same chamber we also have that $H(Wx^*+b)=H(Wu^*+b)=x^*$; thus, $x^*$ is a steady state of the BN.

From Lemma \ref{lemma:conv2H} we can make the norm of $D\Gamma^J$  small so that $u^*$ is the unique fixed point of $\Gamma^J$ in $K\cap \mathcal{C}(x^*)$. Since $\Gamma^J$ converges uniformly to $H(Wx^*+b)=h(x^*)=x^*$ on $K\cap \mathcal{C}(x^*)$, we have that $u^*=\Gamma^J(u^*)$ converges to $x^*$. Finally, to prove that the steady state of the ODE is asymptotically stable, we will show that the Jacobian matrix of the ODE can be seen as a small perturbation of a matrix that has negative eigenvalues. We will use the alternative form of $\Gamma^J_i$:
\[\Gamma^J_i(u)=\frac{2A_i(u)J}{-A_i(u)+I_i(u)+A_i(u)J+I_i(u)J+\sqrt{\Delta_i(u)}}.\] 

Denote $f^J=(f_1^J,\ldots,f_N^J)$ where $f_i^J(u)=A_i(u) \frac{1-u_i}{J+1-u_i}-I_i(u) \frac{u_i}{J+u_i}$. We now compute $Df(u)$. For $i\neq j$ we have 
\[\frac{\partial f_i^J}{\partial u_j}=w_{ij}^+\frac{1-u_i}{J+1-u_i}-w_{ij}^- \frac{u_i}{J+u_i}.\] Also, 
\[\frac{\partial f_i^J}{\partial u_i}=w_{ii}^+\frac{1-u_i}{J+1-u_i}-w_{ii}^- \frac{u_i}{J+u_i}-(w_{ii}^+ u_i +b_i^+)\frac{J}{(J+1-u_i)^2}-(w_{ii}^- u_i +b_i^-)\frac{J}{(J+u_i)^2}.\]

Let $Z^J$ be the matrix given by $Z^J_{ij}=w_{ij}^+\frac{1-u_i}{J+1-u_i}-w_{ij}^- \frac{u_i}{J+u_i}$ and denote with $E^J$ the diagonal matrix with entries $E_{ii}^J=-(w_{ii}^+ u_i +b_i^+)\frac{J}{(J+1-u_i)^2}-(w_{ii}^- u_i +b_i^-)\frac{J}{(J+u_i)^2}$. Then,
\[Df^J(u)=Z^J+E^J\]
 where the entries of $Z^J$ are bounded  and $E^J$ is a diagonal matrix. We now will show that for any steady state of the ODE in $K$, $\lim_{J\rightarrow 0} E^J_{ii}=-\infty$ as $J\rightarrow 0$. After showing this, we can see $Df^J(u^*)$ as a small perturbation of a matrix that has negative eigenvalues. It follows (e.g. using the Gershgorin circle theorem) that the eigenvalues of $Df^J(u^*)$ have negative real part, and hence, $u^*$ is asymptotically stable.

We now show that $\lim_{J\rightarrow 0} E^J_{ii}=-\infty$ as $J\rightarrow 0$.
By computing $\frac{(\Gamma^J_i(u))^2}{J}$ and setting $J=0$, it follows that $\lim_{J\rightarrow 0}\frac{(\Gamma^J_i(u))^2}{J}=0$ when $A_i(u)-I_i(u)$ is negative. Similarly, we obtain that  $\lim_{J\rightarrow 0}\frac{(1-\Gamma^J_i(u))^2}{J}=0$ when $A_i(u)-I_i(u)$ is positive. From these two limits, it follows that if $|A_i(u)-I_i(u)|\neq 0$, then $\lim_{J\rightarrow 0}\frac{(J+\Gamma^J_i(u))^2}{J}=0$ or $\lim_{J\rightarrow 0}\frac{(J+1-\Gamma^J_i(u))^2}{J}=0$. Furthermore, since $A_i(u)-I_i(u)$ is uniformly nonzero on $K$, the convergence is uniform. 

If $u^*\in K$ is a steady state of the ODE, then  $u_i^*=\Gamma^J(u^*)$ and\\
 $\lim_{J\rightarrow 0}\left( -(w_{ii}^+ u^*_i +b_i^+)\frac{J}{(J+1-u^*_i)^2}-(w_{ii}^- u^*_i +b_i^-)\frac{J}{(J+u^*_i)^2}\right)=$\\
$\lim_{J\rightarrow 0} \left( -(w_{ii}^+ u^*_i +b_i^+)\frac{J}{(J+1-\Gamma^J_i(u^*))^2}-(w_{ii}^- u^*_i +b_i^-)\frac{J}{(J+\Gamma^J_i(u^*))^2}\right)=-\infty$. Note that uniform convergence is needed in the last step because $u^*$ depends on $J$.

\end{proof}

\subsection{Proof of Theorems \ref{thm:transition} and \ref{thm:trajectory} }

In the rest of this section, ``ODE'' will refer to the ODE in Eq.~(\ref{eqn:ProblemEquation}) and ``BN'' will refer to the BN in Eq.~(\ref{eqn:mainReduced_BN}). 

Notice that for any $x\in\{0,1\}^N$, $\mathcal{C}(x)= \{u\in[0,1]^N:H(Wu+b)=H(Wx+b)\}$. We now prove Theorem \ref{thm:transition}.

\begin{proof}
Let $y=h(x)$ and for simplicity in the notation, assume that $y=(0,\ldots,0)$. 

In the case $x=y$, we will show that $(0,\ldots,0)$ contains an invariant set for the original ODE. Since $x=(0,\ldots,0)$ and $h(x)=x$, we have that $\mathcal{C}(x)=\cap_{i=1}^N\{u\in[0,1]^N:\sum_{j=1}^N w_{ij} u_j +b_i<0\}$. We now consider a small hypercube of the form $K=[0,\delta]^N$ with $\delta$ small  so that $K\subseteq \mathcal{C}(x)$. We claim that for $J$ small  $K$ is invariant. Since we already showed that $[0,1]^N$ is invariant, it is enough to check that  if $u\in K$ with $u_i=\delta$, then  $f_i^J(u)\leq 0$. Since $K$ is compact and $\sum_{j=1}^N w_{ij} u_j +b_i<0$ for all $i$ and for all $u\in K$, there is $r>0$ such that $\sum_{j=1}^N w_{ij} u_j +b_i\leq -r$ for all $u\in K$. It follows that $A_i(u)-I_i(u) \frac{u_i}{J+u_i}$ converges uniformly to $A_i(u)-I_i(u)=\sum_{j=1}^N w_{ij} u_j +b_i\leq -r$ on  $\{u\in K:u_i=\delta\}$ as $J\rightarrow 0$. Also, $f_i^J(u)=A_i(u) \frac{1-u_i}{J+1-u_i}-I_i(u) \frac{u_i}{J+u_i}\leq A_i(u)-I_i(u) \frac{u_i}{J+u_i}$ on  $\{u\in K:u_i=\delta\}$. Thus, on  $\{u\in K:u_i=\delta\}$, $f_i^J$ is bounded above by a function that converges uniformly to a negative function.  Then, there is $\epsilon_K>0$ such that for all $0<J<\epsilon_K$, $f_i^J$ is negative on  $\{u\in K:u_i=\delta\}$. Then, $K$ is invariant.

In the case $x\neq y$, we assume for simplicity that $x=(1,0,\ldots,0)$ and $y=(0,0,\ldots,0)$. Then, since $h(x)=y$ and $h_1(y)=0$, we have the following
\[\sum_{j=1}^N w_{ij}u_j+b_i<0, \textrm{ for all } u\in \mathcal{C}(x), \textrm{ and } \sum_{j=1}^N w_{1j}u_j+b_i<0, \textrm{ for all } u\in \mathcal{C}(y).\]

In particular, $\sum_{j=1}^N w_{1j}u_j+b_i<0$ for all $u\in \mathcal{C}(x)\cup \mathcal{C}(y)$. This also means that the hyperplane that separates $\mathcal{C}(x)$ and $\mathcal{C}(y)$ is $\{u:\sum_{j=1}^N w_{kj}u_j+b_i=0\}$ for some  $k\neq 1$; then, the common face of $\mathcal{C}(x)$ and $\mathcal{C}(y)$ is given by 
\[\{u\in[0,1]^N: \sum_{j=1}^N w_{ij}u_j+b_i<0 \textrm{ for $i\neq k$ and $\sum_{j=1}^N w_{kj}u_j+b_k=0$}\}.\] 
 Now, for $r>0$ small, we define the set
\[L:=\{u\in[0,1]^N: \sum_{j=1}^N w_{ij}u_j+b_i<-r \textrm{ for $i\neq k$ and $\sum_{j=1}^N w_{kj}u_j+b_k=0$}\}\]
which will be a face of the neighborhood of $x$ that we are looking for (see Fig.~\ref{fig:transition_proof}).
We now project $L$ onto the $u_1=0$ plane (see Fig.~\ref{fig:transition_proof}); that is, define 
\[L_1:=\{(0,u_2,u_3,\ldots,u_N):(u_1,\ldots,u_N\in L \textrm{ for some } u\in L)\}\]
We use $L_1$ to ``generate'' a box parallel to the $u_1$ axis (see Fig.~\ref{fig:transition_proof}); namely, consider
\[B:=\{u\in[0,1]^N: (0,u_2,\ldots,u_N)\in L_1 \}.\]
Now, consider the neighborhood of $x$ given by \[K:=B\cap \mathcal{C}(x).\]
$K$ is a polytope such that $L$ is one of its faces.  Similar to the case $x=y$, there is $\epsilon_K>0$ such that for all $0<J<\epsilon_K$ we have that for any face of $K$ other than $L_1$ the vector field points inward. Also, the first coordinate of the vector field is negative on $K$. Thus, any solution with initial condition  in $K$, must exit $K$ through its face $L_1$ and then enter  $\mathcal{C}(y)$. That is, the ODE transitions from $K$ to $\mathcal{C}(y)$.
\end{proof}

\begin{figure}[h]
\begin{center}
\includegraphics[scale = .6]{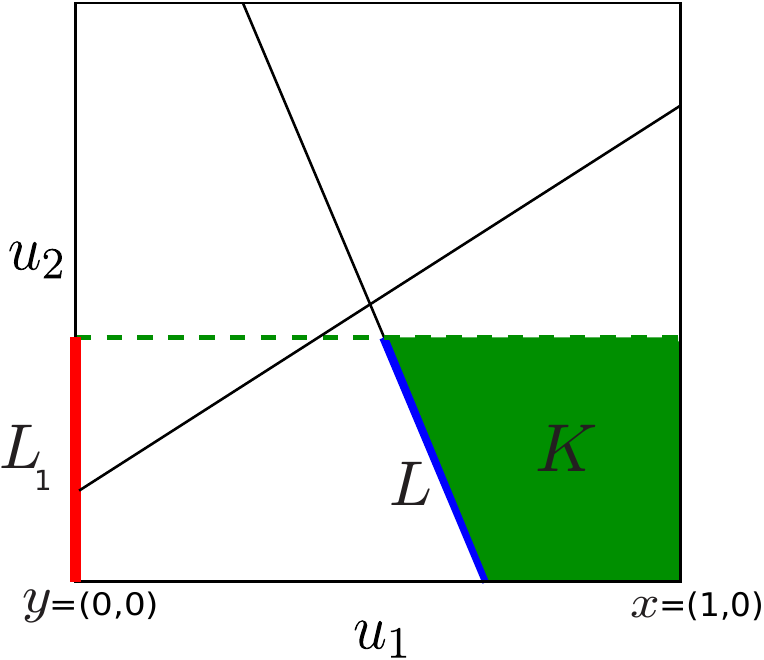}
\parbox{.9\textwidth}{
\caption[]{\footnotesize Sets $L$ (blue), $L_1$ (red), and $K$ (green) in the proof of Theorem \ref{thm:transition} for $N=2$.
\label{fig:transition_proof}}
}
\end{center}
\end{figure}

We prove Theorem \ref{thm:trajectory}.

\begin{proof}
We proceed as in the proof of Theorem \ref{thm:transition} and for simplicity we assume that $x=(1,0,\ldots,0)$ and $y=h(x)=(0,\ldots,0)$. Since $h(x)=y$ and $h_1(y)=0$, we have the following
\[\sum_{j=1}^N w_{ij}u_j+b_i<0, \textrm{ for all } u\in \mathcal{C}(x), \textrm{ and } \sum_{j=1}^N w_{1j}u_j+b_i<0, \textrm{ for all } u\in \mathcal{C}(y).\]

In particular, $\sum_{j=1}^N w_{1j}u_j+b_i<0$ for all $u\in \mathcal{C}(x)\cup \mathcal{C}(y)$. This also means that the hyperplane that separates $\mathcal{C}(x)$ and $\mathcal{C}(y)$ is $\{u:\sum_{j=1}^N w_{kj}u_j+b_i=0\}$ for some  $k\neq 1$; furthermore, since this hyperplane is parallel to the axes, $(w_{k1},w_{k2},\ldots,w_{kN})=(w_{k1},0,\ldots,0)$. Then, the hyperplane that separates $\mathcal{C}(x)$ and $\mathcal{C}(y)$ is $\{u: u_1=\frac{b_i}{-w_{k1}}\}$ and the common face of $\mathcal{C}(x)$ and $\mathcal{C}(y)$ is given by 
\[\{u\in[0,1]^N: \sum_{j=1}^N w_{ij}u_j+b_i<0 \textrm{ for $i\neq k$ and $u_1=\frac{b_i}{-w_{k1}}$}\}.\]

Now, let $K$ be a compact subset of $\mathcal{C}(x)$ and consider $r>0$ small such that\\ $K\subseteq K^0:=\{u\in [0,1]^N:\sum_{j=1}^N w_{ij}u_j+b_i\leq -r \textrm{ for $i\neq k$ and $u_1\geq\frac{b_i}{-w_{k1}}$} \}$ (see Fig.~\ref{fig:trajectory_proof}). Since the hyperplanes are parallel to the axis, $K^0$ is a box with faces parallel to the axes and $K^0$ also shares a face with $\mathcal{C}(y)$. Then, similar to the proof of Theorem \ref{thm:transition}, there is $\epsilon_K>0$ such that for all $0<J<\epsilon_K$ we have that at the faces of $K^0$ other than the shared with $\mathcal{C}(y)$ the vector field of the ODE points inward, and the first entry of the vector field is negative. Then, the ODE will transition from $K^0$ to $\mathcal{C}(y)=\mathcal{C}(h(x))$.

Now, let $K^1$ be a compact subset of $\mathcal{C}(h(x))$ such that $K^1$ intersects all solutions that start in $K$ (see Fig.~\ref{fig:trajectory_proof}). Then, for all $0<J<\epsilon_K$ (making $\epsilon_K$ smaller if necessary) the ODE transitions from $K^1$ to $\mathcal{C}(h^2(x))$. This also means that the ODE transitions from $K^0$ to $\mathcal{C}(h^2(x))$. The proof follows by induction.
\end{proof}

\begin{figure}[h]
\begin{center}
\includegraphics[scale = .6]{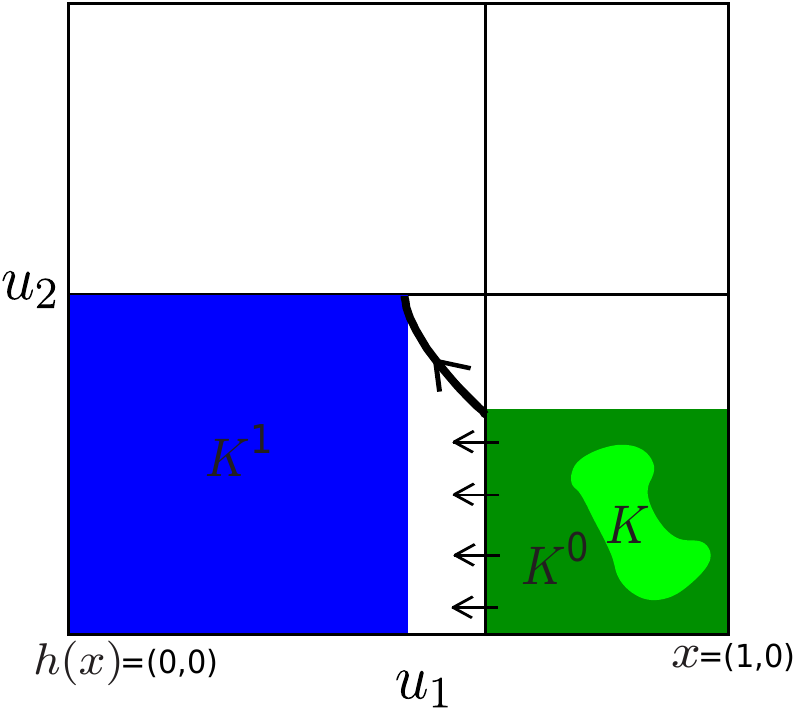}
\parbox{.9\textwidth}{
\caption[]{\footnotesize Sets $K$ (light green), $K^0$ (dark green and light green), and $K^1$ (blue) in the proof of Theorem \ref{thm:trajectory} for $N=2$.
\label{fig:trajectory_proof}}
}
\end{center}
\end{figure}

\end{document}